\documentclass{IEEEtran}

\IEEEoverridecommandlockouts
\usepackage{cite}
\usepackage{amsmath,amssymb,amsfonts}
\usepackage{algorithmic}
\usepackage{graphicx}
\usepackage{textcomp}
\usepackage[table]{xcolor}
\usepackage{acronym}
\usepackage{bm}
\usepackage{amsthm}
\usepackage{enumitem}
\usepackage{subcaption}
\usepackage{multirow}
\usepackage{makecell}

\def\BibTeX{{\rm B\kern-.05em{\sc i\kern-.025em b}\kern-.08em
    T\kern-.1667em\lower.7ex\hbox{E}\kern-.125emX}}

\newcommand{\scaleSection}{\vspace{0cm}}
\newcommand{\scaleSubsection}{\vspace{0cm}}

\newcommand{\scaleSectionBelow}{\vspace{-0cm}}
\newcommand{\scaleSubsectionBelow}{\vspace{-0cm}}

\newtheorem{theorem}{Proposition}
\newtheorem{problem}{Problem}
\newtheorem{subproblem}{Problem}[problem]

\newtheorem{remark}{Remark}

\acrodef{BER}{bit error rate}
\acrodef{pdf}{probability density function}
\acrodef{RX}{receiver}
\acrodef{SMC}{synthetic molecular communication}
\acrodef{VOC}{volatile organic compound}
\acrodef{TX}{transmitter}
\acrodef{ROAMT}[ROADMT]{relative observed abundance of different molecule types}
\acrodef{LC}{low-complexity}
\acrodef{CSK}{concentration shift keying}
\acrodef{STE}{source term estimation}
\acrodef{SM}{single molecule type}
\acrodef{CS}{cumulative score}

\newcommand{\uvec}{\mathbf{u}}
\newcommand{\x}{\mathbf{x}}

\newcommand{\z}{\mathbf{z}}

\newcommand{\E}[1]{\mathrm{E}\left\{ #1 \right\}}

\newcommand{\transpose}{^\mathrm{T}}

\newcommand{\xrx}{\mathbf{x}_{\mathrm{RX}}}
\newcommand{\xtx}{\mathbf{x}_{\mathrm{TX}}}
\newcommand{\rrx}{r_{\mathrm{RX}}}
\newcommand{\rtx}{r_{\mathrm{TX}}}

\newcommand{\Ehat}[2]{\hat{\mathrm{E}}_{#1}\left\{ #2 \right\}}
\newcommand{\uturbulent}{\mathbf{u}_{\mathrm{turb}}}

\begin{document}

\title{Source Distance and Source Strength Estimation in Turbulent Flow: Exploiting Molecule Degradation Diversity}

\author{\IEEEauthorblockN{Bastian Heinlein, Timo Jakumeit, Robert Schober, Maximilian Schäfer, and Vahid Jamali}%
\thanks{Bastian Heinlein, Timo Jakumeit, Robert Schober, and Maximilian Schäfer are with Friedrich-Alexander-Universität Erlangen-Nürnberg, Erlangen, Germany. Bastian Heinlein and Vahid Jamali are also with Technical University of Darmstadt, Darmstadt, Germany.}
\thanks{This work was funded by the Deutsche Forschungsgemeinschaft (DFG, German Research Foundation) – GRK 2950 – Project-ID 509922606 and grant SCHA 2350/4-1. It has also received funding by the German Federal Ministry of Research, Technology and Space (BMFTR) through Project Internet of Bio-Nano-Things (IoBNT) – grant numbers 16KIS1987 and 16KIS1992.}
\thanks{This work has been accepted in part for presentation at the IEEE Global Communications Conference, 2026~\cite{heinlein:globecom}.}
}

\maketitle
\begin{abstract}
Locating an odor source and estimating its strength are central problems in both nature and engineering. 
For example, an insect may decide whether a food source is worth traveling to based on its distance and how many molecules it emits. 
Similarly, a drone or a sensor network may need to locate a pipe leakage and quantify how much chemical is released. Such tasks are highly challenging due to the complex propagation dynamics of turbulent flow. 
While source distance estimation and source strength estimation have been extensively investigated for sources releasing a single molecule type, in both turbulent and non-turbulent flow, many odor sources release mixtures of multiple molecule types. 
In this study, we investigate whether considering more than a single molecule type can aid source distance and source strength estimation in challenging turbulent flow environments. 
We propose to exploit that different molecule types in a mixture degrade at different rates in the atmosphere. Consequently, the relative abundance of the different molecule types observed at a receiver provides information about the source distance that a single molecule type does not.
Since this relative abundance is independent of the source strength, it resolves the ambiguity between distance and strength inherent in the observed signal intensity and thus also improves strength estimation. 
Building on this idea, we propose both heuristic low-complexity and learning-based schemes that combine information encoded in the relative abundance of different molecule types with established features from the literature such as mean concentration and whiff duration. The performance of all proposed schemes is validated using numerical simulations of molecule propagation in turbulent flow.
The proposed low-complexity schemes can achieve competitive estimation quality compared to learning-based schemes that rely only on a single molecule type, while the proposed learning-based schemes outperform their single-type counterparts.
Molecule mixtures can therefore support goal-oriented tasks such as odor source distance and strength estimation even under challenging turbulent propagation conditions.
\end{abstract}

\acresetall
\scaleSection\section{Introduction}\scaleSectionBelow
In nature, olfactory cues are utilized by many species for various tasks such as attracting mating partners, searching for food, or sensing danger~\cite{szyszka:extracting_spatial_information_temporal_odor_patterns,david:finding_sex_pheromone_source_gypsy_moths,martin:osmotropotaxis_honey_bee}. 
Inspired by this natural molecule-based communication and sensing, increasing attention has been directed towards building engineered systems that exploit these principles.
These so-called \ac{SMC} systems are envisioned to enable novel applications, e.g., for pest control or leakage detection~\cite{vakilipoor:mc_agriculture,nakano:mcn_opportunities_challenges}. 
In most natural molecule-based communication systems, but also in the envisioned air-borne \ac{SMC} systems, molecules propagate through turbulent airflow, especially over longer distances. For example, female moths release pheromones to attract mating partners that are possibly dozens of meters away~\cite{david:finding_sex_pheromone_source_gypsy_moths} and bees rely on olfactory cues that propagate through turbulent airflow for food search~\cite{martin:osmotropotaxis_honey_bee}. The same holds for \ac{SMC} applications: for example, sensor networks installed to detect and localize gas leaks operate outdoors where the released molecules are transported over tens of meters by turbulent airflow~\cite{travis:neural_network_locate_quantify_gas_leaks}.

While a few studies have considered turbulent flow conditions in specific geometries like pipes or for specific 
applications such as the propagation of virus particles~\cite{abbaszadeh:kolmogorov_turbulence_information_dissipation_mc,
abbaszadeh:mi_noise_distributions_molecular_signals_using_lif,
gulec:computational_approach_characterization_airborne_pathogen_transmission,
ozmen:high_speed_chemical_vapor_communication_pid}, the majority of \ac{SMC} research has 
focused on propagation conditions dominated by diffusion or uniform/laminar 
flow~\cite{lotter:experimental_research_smc,jamali:olfaction_inspired_MC,
mcguiness:experimental_analytical_analysis_macroscale_mc_closed_boundaries,
qiu:mc_link_monitoring_confined_environments}. 
This is especially the case for the extensive \ac{SMC} literature on parameter estimation, where 
a wide range of schemes has been proposed to estimate quantities such as the source strength or channel parameters like the distance between 
\ac{TX} and \ac{RX} or the flow velocity~\cite{noel:joint_channel_parameter_estimation_diffusive_mc,sahoo:channel_parameter_estimation,cheng:deep_learning_based_estimation_distance_emission_time_mobile_mc,wang:distance_estimation_schemes_diffusion_based_MC}. 
However, these schemes usually assume analytical channel models and/or an impulsive molecule release and are therefore, like most existing work on \ac{SMC} system design~\cite{mcguiness:experimental_analytical_analysis_macroscale_mc_closed_boundaries,heinlein:tmbmc2026}, not directly 
applicable to turbulent flow conditions. 
How molecules propagate in turbulent airflow, and how this should be accounted for when 
estimating parameters and designing \ac{SMC} systems, therefore remains a largely open research 
question.

In other scientific communities, \textbf{\ac{STE}}, a problem conceptually similar to parameter estimation, has already gathered tremendous interest and has also been investigated for complex propagation environments~\cite{hutchinson:review_ste_methods}. 
In \ac{STE}, the goal of a static or moving \ac{RX} is to determine properties of a molecule source\footnote{In the remainder of this paper, we will use the terms \textit{source} and \textit{\ac{TX}} interchangeably.}, such as its location or strength, where the latter refers to how many molecules are released per unit time. 
A large part of the \ac{STE} literature considers active searchers, i.e., agents that can move in the environment and that usually aim to reach the molecule source itself. Such approaches, e.g., based on \textit{infotaxis} or reinforcement learning, can explain insect behavior and have inspired robotic systems for source localization~\cite{vergassola:infotaxis_strategy_searching_without_gradients,loisy:searching_for_source_without_gradients,vanhove:actively_inferring_methane_sources_drones,hutchinson:entrotaxis_strategy_autonomous_search_source_reconstruction,abughanam:performance_altruistic_infotaxis} that have also been partially validated experimentally~\cite{hutchinson:source_term_estimation_hazardous_airborne_release,park:autonomous_source_search_for_uavs,jin:cumulative_information_path_planning}. Similarly, approaches for source localization and strength estimation based on individual static sensors or sensor networks have been investigated using both simulations and real-world measurements~\cite{gulec:localization_passive_source_sensor_network,balocchi:enhanced_gas_source_localization_distributed_iot_sensors,travis:neural_network_locate_quantify_gas_leaks,cerizza:reconstruction_scalar_source_intensity,ryan:uncertainty_quantification_steady_transient_source_term_estimation_urban_environment}. 
However, the majority of the aforementioned literature developed algorithms and schemes are too complex to permit a systematic investigation of which features of the olfactory cues (e.g., the mean concentration or the time between concentration peaks) are actually informative for \ac{STE}. 
Yet, exactly this fundamental understanding is required to build reliable and robust \ac{SMC} systems. 
Therefore, it can be beneficial to study instead the simpler problem of estimating the distance between a (static) \ac{RX} and a molecule source. This is a central aspect of source localization and at the same time more tractable compared to exact source localization with moving agents. Source distance estimation has received increasing attention and has been investigated using both simulated and experimental data~\cite{rigolli:learning_predict_target_location_turbulent_odor_plumes,gulec:distance_estimation_methods_practical_macroscale_mc_system,gulec:fluid_dynamics_based_distance_estimation_algorithm}.
In particular, the authors of~\cite{rigolli:learning_predict_target_location_turbulent_odor_plumes} used high-resolution fluid dynamic simulations and machine learning models to systematically investigate which features of molecule counts observed at an \ac{RX} are most predictive of its distance to the source. 
Like all above-mentioned studies on \ac{STE}, only a source emitting a single molecule type was considered in~\cite{rigolli:learning_predict_target_location_turbulent_odor_plumes}.

On the other hand, recent studies in the \ac{SMC} literature indicate that systems employing molecule mixtures for communication can achieve a higher information content and improved robustness compared to systems relying only on a single molecule type~\cite{jamali:olfaction_inspired_MC,araz:rskm_time_varying_MC_channels}. 
Although these contributions consider settings with simplified, laminar flow conditions, it is known that natural molecule-based communication and sensing systems, which usually operate under turbulent conditions, also exploit molecule mixtures~\cite{Schuman2023,renou:insect_olfactory_communication_complex_changing_world}. 
These results show that mixtures provide additional degrees of freedom at the \ac{RX} that are unavailable when relying on a single molecule type.
Motivated by this, it is natural to investigate whether and how molecule mixtures can also aid \ac{STE} and consequently be beneficial for the development of \ac{SMC} systems, especially in scenarios with more realistic turbulent airflow. 

In this paper, we focus on the intertwined problems of \textbf{source distance estimation} and \textbf{source strength estimation}. For sources emitting a single molecule type, these two problems are strongly coupled because the features most informative about the source distance depend on the molecule concentrations observed at the \ac{RX}~\cite{rigolli:learning_predict_target_location_turbulent_odor_plumes}. However, the observed molecule concentrations are also modulated by the strength of the molecule source, introducing ambiguity in settings where both the source distance and the source strength are unknown. 

We propose to address this ambiguity by exploiting that many molecule types are subject to atmospheric degradation when they travel through the air where they interact with atmospheric ozone or UV light. 
Crucially, the speed of this degradation differs between molecule types~\cite{williams:human_odour_thresholds_tuned_atmospheric_chemical_lifetimes}, a property we refer to as \textbf{molecule degradation diversity} in the following.
This phenomenon is expected to occur widely, since many natural molecule sources release different types of molecules~\cite{doetterl:chemistry_biosynthesis_biology_floral_volatiles,ando:semiochemicals_sex_pheromones} and additional molecule types can be added to synthetic sources. 
Molecule degradation diversity implies that the composition of a released molecule mixture changes as the mixture travels through the air: molecule types with higher degradation rates are depleted faster than those with lower rates. 
Therefore, the mixture composition observed at the \ac{RX} indicates the travel time and, in turn, the expected distance  between \ac{TX} and \ac{RX}. 

For clarity, we focus here on mixtures with two molecule types\footnote{We note that our proposed approaches can be extended to mixtures with more than two molecule types, which is left for future work.} and assume that one of these types is subject to atmospheric degradation while the other is not. 
Then, the mixture composition observed at the \ac{RX} is fully described by the \textit{\ac{ROAMT}}, i.e., the ratio between the observed numbers of degraded and non-degraded molecules.

\ac{ROAMT} complements the features already proposed for single molecule type sources in~\cite{rigolli:learning_predict_target_location_turbulent_odor_plumes} and has various properties that make it a promising feature for \ac{STE}: First of all, degradation rates of different molecule types can vary over several orders of magnitude~\cite{williams:human_odour_thresholds_tuned_atmospheric_chemical_lifetimes,koppmann:chemistry_vocs_atmosphere,atkinson:atmospheric_degradation_vocs}, making \ac{ROAMT} suitable for distance estimation over a large range of distances. 
Second, it enables a simpler estimation of the source distance compared to the single molecule type-based features proposed in~\cite{rigolli:learning_predict_target_location_turbulent_odor_plumes} due to a straightforward relationship between molecule degradation and average travel time. 
Finally, \ac{ROAMT} is, per definition, independent of the source strength as long as the release ratio of the different molecule types remains constant. Thus, it addresses the ambiguity of concentration-based single molecule type features in settings with unknown source strength and thereby can aid also source strength estimation.

\begin{figure*}
    \centering
    \includegraphics[width=0.8\linewidth]{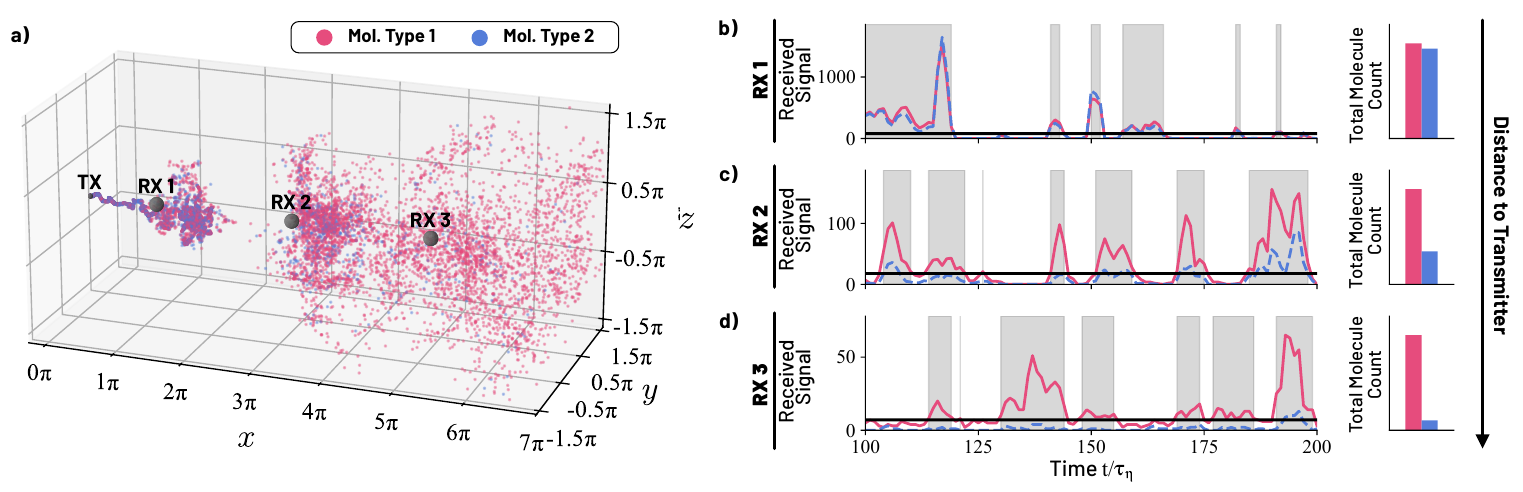}
    \caption{\textbf{System Model.} a) Molecules released by a small spherical \ac{TX} propagate through complex turbulent airflow, in which transparent \acp{RX} are positioned. b)-d) Each transparent \ac{RX} records one time-series per molecule type, where time-spans with high molecule concentrations are called \textit{whiffs} (grey). The horizontal black lines indicate the threshold above which a signal is considered a whiff. The ratio of observed molecules at an \ac{RX} changes with the distance to the source.}
    \label{fig:system_model}
\end{figure*}

To the best of the authors' knowledge, exploiting molecule degradation diversity for source distance and source strength estimation in turbulent conditions has not been studied in the literature before. The main contributions of this paper can be summarized as follows: 
\begin{enumerate}
    \item We propose a novel \ac{LC} distance estimator based on \ac{ROAMT} for the case of two molecule types, one of which is degraded in the atmosphere. Extending the conference version of this paper~\cite{heinlein:globecom}, we develop a new \ac{LC} source strength estimator. Both proposed estimators rely on knowledge of only a few parameters and do not require extensive training data, enabling \ac{STE} in settings where it is impractical to obtain enough data to train machine learning models. 
    \item Furthermore, we propose a learning-based estimator for the source distance based on \ac{ROAMT}. Additionally, we propose to augment this estimator with the features employed for single molecule type sources in~\cite{rigolli:learning_predict_target_location_turbulent_odor_plumes} for even higher estimation accuracy. Compared to~\cite{heinlein:globecom}, we add a more detailed analysis of the influence of the system parameters on estimation accuracy and evaluate it also in settings where the source strength is unknown. 
    \item Extending~\cite{heinlein:globecom}, we propose and analyze a learning-based estimator to estimate the source strength based on \ac{ROAMT} and the features used for single molecule type sources in~\cite{rigolli:learning_predict_target_location_turbulent_odor_plumes}. 
    \item We leverage a dataset from~\cite{biferale:turb_smoke_database_lagrangian_pollutants} containing the trajectories of several million molecules propagating through complex turbulent flow to analyze all proposed estimators. In particular, we systematically investigate which individual features and which feature combinations are informative for source distance estimation and source strength estimation, respectively. 
    \item Using the aforementioned dataset, we demonstrate that the proposed \ac{LC} estimators leveraging \ac{ROAMT} can approach the performance of the corresponding learning-based single molecule type-based schemes. Furthermore, we show that learning-based schemes combining \ac{ROAMT} and single molecule type features outperform single molecule type-based schemes. 
\end{enumerate}

The remainder of this paper is structured as follows: In Section~\ref{sec:system_model}, we introduce the employed system model, before reporting the proposed source distance estimation schemes in Section~\ref{sec:schemes}. In Section~\ref{sec:strength}, we derive the proposed source strength estimation schemes. We then evaluate all proposed schemes in Section~\ref{sec:eval}, before summarizing our results and pointing out directions for future research in Section~\ref{sec:conclusion}.

\textbf{Notation}: Vectors and matrices are denoted by lowercase and uppercase bold letters, respectively. $\nabla$ and $\Delta$ denote respectively the Nabla and the Laplace operator. $\frac{\mathrm{d}}{\mathrm{d}t}$ and $\partial_t$ denote respectively the derivative and partial derivative with respect to time $t$. $\partial_l \, y[l]$ applies the difference operator to a discrete-time signal $y[l]$, where $l$ is the time index.
Sets are shown by calligraphic letters and the cardinality of set $\mathcal{A}$ is denoted by $|\mathcal{A}|$. A random variable uniformly distributed over a set $\mathcal{A}$ is denoted by $\mathcal{U}(\mathcal{A})$.
$\Ehat{l \cdot T_\mathrm{s} \in \mathcal{T}}{y[l]}$ denotes the empirical mean over a continuous-time interval $\mathcal{T}$ of a discrete-time signal $y[l]$, which has been sampled with sampling interval $T_\mathrm{s}$. Finally, $\boldsymbol{1}\{ \cdot \}$ denotes the indicator function.

\scaleSection\section{System Model}\scaleSectionBelow\label{sec:system_model}
In this work, we study a system to investigate source distance estimation and source strength estimation in turbulent airflow as illustrated in Figure~\ref{fig:system_model}\footnote{Note that, in this paper, we use the dataset from~\cite{biferale:turb_smoke_database_lagrangian_pollutants}, which employs normalized lengths and time scales, resulting in lengths that are multiples of $\pi$ and times that are multiples of the Kolmogorov timescale $\tau_\eta$.}. In the considered system, a \ac{TX} continuously releases up to two different types of molecules\footnote{The proposed schemes would also be applicable to \acp{TX} releasing more than two molecule types. For example, they could simply ignore all but two molecule types, or be applied to all pairs of molecule types in the released mixture and the results could be combined.}, which propagate through turbulent flow. Transparent \acp{RX} (e.g., \ac{RX} 1-\ac{RX} 3 in Figure~\ref{fig:system_model}) record the observed number of molecules over time and use these signals for \ac{STE}. In the following, we introduce the individual components of the system model in more detail. 

\scaleSubsection\subsection{Transmitter}\scaleSubsectionBelow\label{sec:system_model:tx}
We consider a small spherical \ac{TX} at position $\xtx = [0, 0, 0]\transpose$ with radius $\rtx$. 
The \ac{TX} releases a mixture of two molecule types at a total flux strength $J$, i.e., $J$ molecules are released per unit time. 
A fraction $p_1$ of the released molecules is of type 1 and the remaining fraction $p_2 = 1-p_1$ is of type 2, so that the two types are released at fluxes $p_1 \cdot J$ and $p_2 \cdot J$, respectively. 
We assume that the total flux strength takes one of $M$ possible values $J^{(m)} \in \mathcal{J}$, where $\mathcal{J} = \{J^{(1)}, \dots, J^{(M)}\}$, and denote by $m$ the index of the flux strength at which the \ac{TX} currently releases molecules, i.e., $J=J^{(m)}$. 
Consequently, the flux strengths of type 1 and type 2 are denoted respectively by $J^{(m)}_1$ and $J^{(m)}_2$.
For simplicity, we assume that all total flux strengths fulfill $J^{(m)} = J^{(1)} + \frac{J^{(M)}-J^{(1)}}{M-1} \cdot (m-1)$, $m \in \{1, \dots, M\}$.

\begin{remark}
We use discrete levels of source strengths which correspond to a \ac{CSK} scheme with equidistant modulation levels. At the same time, as $M$ grows large for a fixed $|J^{(M)}-J^{(1)}|$, this scheme approximates a continuous range of source strengths.
\end{remark}

\scaleSubsection\subsection{Channel with Turbulent Flow}\scaleSubsectionBelow\label{sec:system_model:channel}
We consider a channel with turbulent flow through which the molecules propagate as passive \textit{tracers}, i.e., they do not change the underlying velocity field. 
In particular, the molecules' propagation is governed by a velocity field $\uvec(\x,t)$, where $\x$ and $t$ denote respectively the spatial coordinates and time. 
According to the \textit{Reynold decomposition}~\cite{pope:turbulent_flows}, the velocity field $\uvec(\x,t)$ can generally be written as the sum of a turbulent component $\uturbulent(\x,t)$ and a mean component $\uvec_0$, i.e., 
\begin{equation}
    \uvec(\x,t) = \uturbulent(\x,t) + \uvec_0. 
\end{equation}
Following~\cite{biferale:turb_smoke_database_lagrangian_pollutants}, whose dataset we use for the evaluation, we consider in this work a time- and space-independent average component $\uvec_0 = [u, 0, 0]\transpose$, where $u$ denotes the wind speed in $x$-direction. 
Here, $u$ is approximately twice as large as the typical turbulent fluctuation speed (see \cite[Table~I]{biferale:turb_smoke_database_lagrangian_pollutants}). This results in a \textit{cone of detection}, i.e., a volume in which most molecules exist that is longer than wide, as is typically observed in turbulent flow conditions with mean-wind component, see also~\cite{rigolli:learning_predict_target_location_turbulent_odor_plumes,celani:odor_landscapes_turbulent_environments}. 

The turbulent component $\uturbulent(\x,t)$ is computed via direct numerical simulation of the incompressible Navier-Stokes equations, which are given by a momentum equation and a continuity equation as follows~\cite{biferale:turb_smoke_database_lagrangian_pollutants}:
\begin{subequations}
\begin{align}
    \partial_t \uturbulent\! +\! ( \uturbulent\! \cdot\! \nabla) \uturbulent\! &=\! - \nabla P/\rho_0\! + \nu \Delta \uturbulent \!+\! \mathbf{f},\label{eq:system:navier_stokes} \\
    \nabla \cdot \uturbulent &= 0.\label{eq:system:navier_stokes:divergence-free}%
\end{align}
\end{subequations}
Here, $P$, $\rho_0$, $\nu$, and $\mathbf{f}$ denote respectively the pressure, the mean fluid density, the kinematic viscosity, and a homogeneous, isotropic force driving the system to a non-equilibrium state. Eq.~\eqref{eq:system:navier_stokes:divergence-free} enforces incompressibility. 

Because the molecules are modeled as passive tracers that are propagated by the velocity field, the trajectory of the $k$-th molecule of type $s \in \{1,2\}$ can be computed according to
\begin{equation}\label{eq:system:dx}
    \frac{\mathrm{d}}{\mathrm{d}t} \x_{s,k}(t) = \uvec(\x_{s,k}(t), t).
\end{equation}

Besides molecule propagation, we also consider chemical degradation: 
In the atmosphere, molecules are commonly degraded due to effects like \textit{(photo)chemical oxidation}, where molecules are oxidized due to exposure to, e.g., hydroxyl radicals ($\mathrm{OH^-}$)~\cite{koppmann:chemistry_vocs_atmosphere,williams:human_odour_thresholds_tuned_atmospheric_chemical_lifetimes}. 
However, the rate constant of this degradation varies, depending on the molecule type, over multiple orders of magnitude, resulting in \textit{atmospheric lifetimes} between seconds and hours for a given concentration of $\mathrm{OH^-}$~\cite{koppmann:chemistry_vocs_atmosphere,williams:human_odour_thresholds_tuned_atmospheric_chemical_lifetimes,atkinson:atmospheric_degradation_vocs}. 

For our evaluations, we rely on the simulation data from~\cite{biferale:turb_smoke_database_lagrangian_pollutants}, where the position of each molecule is recorded after every Kolmogorov time scale $\tau_{\eta}$, the shortest relevant time scale for turbulent flow. Since molecule degradation can be modeled as a first-order reaction~\cite{williams:human_odour_thresholds_tuned_atmospheric_chemical_lifetimes}, the degradation can be described by the probability that a molecule of type $s \in \{1,2\}$ is degraded between two simulation steps, $p_{\mathrm{deg},s}$.
Therefore, the probability that the $k$-th molecule of species $s$ is still not degraded $l'$ timesteps after molecule release is given as follows 
\begin{equation}\label{eq:survival_probability}
    \Pr[\text{molecule }k\text{ not degraded}|l'] = (1-p_{\mathrm{deg},s})^{l'}.
\end{equation}

\scaleSubsection\subsection{Receiver}\scaleSubsectionBelow\label{sec:system_model:rx}
Since we focus on the effect of turbulent flow in this paper, we assume a transparent \ac{RX} at position $\xrx$ with radius $\rrx$, which can perfectly count the number of molecules of each type within its volume with sampling rate $1/T_\mathrm{s}$. Formally, the \ac{RX} acquires a time series $y_s[l]$, $s \in \{1,2\}$, according to 
\begin{equation}
    y_s[l] = \sum_{k=1}^{k^{\max}_s[l]} \boldsymbol{1}\left\{ \lVert \x_{s,k}(l \cdot T_\mathrm{s}) - \xrx \Vert_2 \leq \rrx \right\},
\end{equation}
where $k^{\max}_s[l]$ denotes the number of molecules of type $s$ that exist in sampling interval $l$, $\boldsymbol{1}\left\{ \cdot \right\}$ the indicator function, and $\lVert \cdot \rVert_2$ denotes the $\ell_2$ norm. Note that we assume in the following always $T_\mathrm{s} = \tau_\eta$ as the dataset used for evaluation (see Section~\ref{sec:eval:setup}) contains particle positions every $\tau_\eta$.

\begin{remark}
    There are practical sampling devices that can distinguish different molecule types even at very low concentrations and high sampling rates. For example, Proton-transfer-reaction mass spectrometry or Raman spectroscopy enable molecule species-resolution at very high sampling rates~\cite{sieburg:monitoring_gas_composition_lab_biogas_plant,yuan:ptr-ms_applications}.
\end{remark}
\begin{remark}
    Note that we fix the \ac{TX} position at $\xtx = [0,0,0]\transpose$ (see Section~\ref{sec:system_model:tx}) and vary the \ac{RX} position $\xrx$ in the simulation domain to obtain different \ac{TX}-\ac{RX} constellations. This is without loss of generality since the turbulent velocity field $\uturbulent(\x,t)$ in~\cite{biferale:turb_smoke_database_lagrangian_pollutants} is homogeneous and the mean wind $\uvec_0$ is space-independent. Thus, the observed signals $y_s[l]$ depend on the \ac{TX} and \ac{RX} positions only through the 
    displacement vector $\xtx-\xrx$. Varying $\xrx$ for a fixed \ac{TX} is therefore equivalent 
    to varying $\xtx$ for a fixed \ac{RX}.
\end{remark}

\scaleSection\section{Distance Estimation Schemes}\scaleSectionBelow\label{sec:schemes}
In the following, we formalize the two considered problem types for source distance estimation and highlight the challenges introduced by turbulent flow before proposing an \ac{LC} distance estimator in Section~\ref{sec:schemes:lc} and a learning-based distance estimator in Section~\ref{sec:schemes:learning}.

\stepcounter{problem}
\begin{subproblem}[Source Distance Estimation with Known Source Strength]\label{prob:distance:known_strength}
    We consider an \ac{RX}, as defined in Section~\ref{sec:system_model:rx}, that remains at a static, but unknown, position $\xrx$ downwind of the unknown source position $\xtx$. The source releases molecules with a fixed and known overall flux strength $J$, i.e., $|\mathcal{J}|=1$. The \ac{RX} then aims to estimate the distance $d$ between $\xrx$ and $\xtx$, i.e., 
    \vspace*{-2mm}\begin{equation}
        d = \lVert \xtx-\xrx \rVert_2.%
    \end{equation}
    For the estimation task, the \ac{RX} utilizes the observed numbers of type 1 and type 2 molecules, i.e., $y_1[l]$ and $y_2[l]$, that are sampled during some observation interval $\mathcal{T} = [t_0, t_0+T]$, where $t_0$ and $T$ denote respectively the start and the duration of the observation interval in continuous time. 
\end{subproblem}

\begin{subproblem}[Source Distance Estimation with Unknown Source Strength]\label{prob:distance:unknown_strength}
    This problem is identical to Problem~\ref{prob:distance:known_strength} with the only difference that $|\mathcal{J}|>1$, i.e., different source flux strengths are possible, while $p_1/p_2$, i.e., the ratio of type 1 and type 2 molecules, is known. 
\end{subproblem}

We note that the existing schemes in the \ac{SMC} literature for the estimation of the distance between \ac{TX} and \ac{RX} (e.g.,~\cite{noel:joint_channel_parameter_estimation_diffusive_mc,sahoo:channel_parameter_estimation,cheng:deep_learning_based_estimation_distance_emission_time_mobile_mc}) are not directly applicable to systems with turbulent flow conditions. This becomes apparent from the received signals in Figure~\ref{fig:system_model}-b)-d). 
Despite the constant release rate of the \ac{TX}, the received signals at the different \acp{RX} are strongly stochastic and cannot be characterized by the typical impulse responses of advection-diffusion channels with laminar flow.
Instead, turbulent flow causes \textit{intermittent} signals, i.e., signals where periods of strong signals (called \textit{whiffs}) and periods of weak signals (called \textit{blanks}) alternate. In Figure~\ref{fig:system_model}-b)-d), the whiffs and blanks of $y_1[l]$ are indicated respectively by gray and white background. Following~\cite{rigolli:learning_predict_target_location_turbulent_odor_plumes}, we formally define a whiff as a time period in the observation interval $\mathcal{T}$, for which
\begin{equation}
    y_s[l]\! >\! y_{\mathrm{whiff},s}\! =\! W \cdot \Ehat{l \cdot T_\mathrm{s} \in \mathcal{T}}{y_s[l]|y_s[l]\!>\!0} \;,\; s\! \in\! \{1,\! 2\},
\end{equation}
is valid. Here, $y_{\mathrm{whiff},s} = W \cdot \Ehat{l \cdot T_\mathrm{s} \in \mathcal{T}}{y_s[l]|y_s[l]>0}$ is the empirical threshold above which the observed signal is considered a \textit{whiff} according to~\cite{rigolli:learning_predict_target_location_turbulent_odor_plumes} and $W > 0$ is an empirically chosen constant. 
Since it has been shown in~\cite{rigolli:learning_predict_target_location_turbulent_odor_plumes} that the choice of $W$ influences the system performance only negligibly, we choose the default value for $W$ identical to the default value in~\cite{rigolli:learning_predict_target_location_turbulent_odor_plumes}, $W=\frac{1}{2}$, in the remainder of this work.

\scaleSubsection\subsection{LC Degradation-enabled Distance Estimation}\scaleSubsectionBelow\label{sec:schemes:lc}
\begin{figure}
    \centering
    \includegraphics[width=0.8\linewidth]{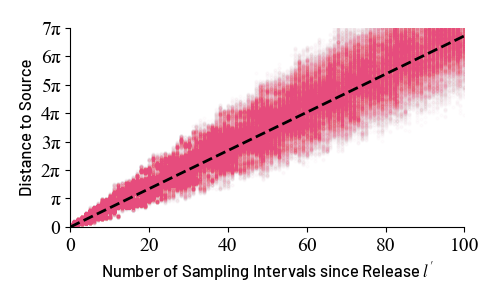}
    \caption{\textbf{Relation between distance to source and travel time.} Plotting for each molecule its distance to the source as a function of the time since its release (red dots), given in terms of the number of sampling intervals, reveals on average a linear relationship (black dashed line) between travel time and source distance.}
    \label{fig:schemes:distance_travel_time}
\end{figure}
As described in Section~\ref{sec:system_model}, different molecule types are subject to different rates of atmospheric degradation. This difference can be exploited for source distance estimation when a source releases a mixture of molecules. 
In this work, we assume that the source releases two types of molecules, where type $1$ is not affected by atmospheric degradation while type $2$ is subject to degradation with degradation probability $p_{\mathrm{deg},2} > 0$. 

Exploiting~\eqref{eq:survival_probability}, the \textit{expected} \ac{ROAMT}, $l'$ sampling intervals after the molecule release, is given as follows
\begin{equation}\label{eq:schemes:lc:exp_mol_ratio}
    \bar{r}(l') = \frac{p_2 \cdot (1-p_{\mathrm{deg},2})^{l'}}{p_1}.
\end{equation}
This expression reveals that $\bar{r}(l')$ can be used to estimate $l'$. 
In turn, the \textit{average} distance to the source, $d$, of a molecule observed $l'$ sampling intervals after its release can be approximated as  follows
\begin{equation}\label{eq:schemes:lc:d}
    d \approx v \cdot l', \vspace*{-2mm}
\end{equation}
where $v$ is a scaling factor related to the mean velocity $u$ in the channel. This dependence is visualized in Figure~\ref{fig:schemes:distance_travel_time}, where based on a simulation of turbulent flow~\cite{biferale:turb_smoke_database_lagrangian_pollutants}, we show the distance of molecules to the source, $d$, as a function of the number of sampling intervals since their release (red dots). Clearly, the average relationship can be approximated by a linear function (dashed black line), confirming the validity of \eqref{eq:schemes:lc:d}. 

\begin{theorem}\label{prop:distance:lc}
Based on~\eqref{eq:schemes:lc:exp_mol_ratio} and~\eqref{eq:schemes:lc:d}, the distance between source and \ac{RX} can be estimated using the following \ac{LC} estimator:
\begin{equation}\label{eq:schemes:lc:estimator}
    \hat{d}^{\mathrm{LC}}(r_{\mathrm{obs}}) = \left[v\cdot\frac{\log(r_{\mathrm{obs}}) - \log(r_0) }{\log(1-p_{\mathrm{deg},2})}\right]_0^{d_{\max}},
\end{equation}
where $[\cdot]_a^b = \max(a, \min(b, \cdot))$ and $r_0=p_2/p_1$. Furthermore, $d_{\max}$ and $r_{\mathrm{obs}}$ denote respectively the largest considered source distance and the empirical approximation of $\bar{r}(l')$ according to 
\begin{equation}\label{eq:schemes:lc:robs}
    r_{\mathrm{obs}} = \frac{ \Ehat{l \cdot T_\mathrm{s} \in \mathcal{T}}{y_2[l]} + \epsilon }{ \Ehat{l \cdot T_\mathrm{s} \in \mathcal{T}}{y_1[l]} + \epsilon },
\end{equation}
where a small $\epsilon$, e.g., $\epsilon = 10^{-2}$, is added for numerical stability. 
\end{theorem}
\begin{proof}
    To obtain the distance estimator in \eqref{eq:schemes:lc:estimator}, we set $r_0=p_2/p_1$ in~\eqref{eq:schemes:lc:exp_mol_ratio}. Additionally, we set $\bar{r}(l') \stackrel{!}{=} r_{\mathrm{obs}}$ and, exploiting~\eqref{eq:schemes:lc:d}, $l'=d\cdot v^{-1}$, resulting in 
    \begin{equation}\label{eq:schemes:lc:robs_proof}
        r_{\mathrm{obs}} \stackrel{!}{=} r_0 (1-p_{\mathrm{deg},2})^{d \cdot v^{-1}}.
    \end{equation}
    Then, we solve for $d = \hat{d}^{\mathrm{LC}}$ by first dividing by $r_0$ before taking the logarithm on both sides, yielding 
    \begin{equation*}
        \log(r_{\mathrm{obs}}) - \log(r_0) = \hat{d}^{\mathrm{LC}} \cdot v^{-1} \log (1-p_{\mathrm{deg},2}).
    \end{equation*}
    Dividing by $v^{-1} \log (1-p_{\mathrm{deg},2})$ yields the final result within the $\min-\max$-expression in \eqref{eq:schemes:lc:estimator}. 
    The $\min$ and $\max$ operations incorporate the prior knowledge about the largest source distance of interest, e.g., the downwind extent of the simulation domain, and the fact that there are no negative distances, respectively.
\end{proof}

\begin{remark}
    Since the estimator introduced in Proposition~\ref{prop:distance:lc} does not require knowledge of $J$, it is directly applicable to both Problem~\ref{prob:distance:known_strength} and~\ref{prob:distance:unknown_strength}.
\end{remark}

\begin{remark}
    The \ac{LC} distance estimator requires only knowledge of $r_0$, $p_{\mathrm{deg},2}$, and $v$. Here, $r_0$ and $p_{\mathrm{deg},2}$ are inherent properties of the source, which can be determined for synthetic sources and be measured in the lab for natural sources. Only $v$ depends on the propagation environment and needs to be estimated in practice. While a corresponding parameter estimation scheme is out of the scope of this work, we would like to note that $v$ could in principle be estimated by measuring $r_{\mathrm{obs}}$ at known distances $d$ from the source and using~\eqref{eq:schemes:lc:robs_proof} to derive an estimator for $v$. Note that~\eqref{eq:schemes:lc:robs_proof} does not depend on $J$ and therefore the above-outlined estimation scheme can be used even when the source strength is unknown.
\end{remark}

\begin{figure*}[t!]
    \centering
    \begin{subfigure}[t]{0.3\textwidth}
        \centering
        \includegraphics[width=\linewidth]{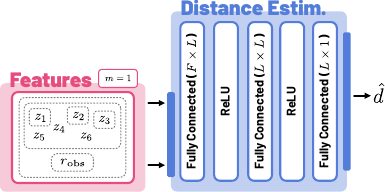}
        \caption{Problem 1a.}
        \label{fig:architecture:1a}
    \end{subfigure}%
    \hspace*{10mm} 
    \begin{subfigure}[t]{0.3\textwidth}
        \centering
        \includegraphics[width=\linewidth]{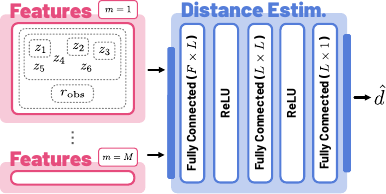}
        \caption{Problem 1b.}
        \label{fig:architecture:1b}
    \end{subfigure}%
    \vspace*{5mm}
    \begin{subfigure}[t]{0.3\textwidth}
        \centering
        \includegraphics[width=\linewidth]{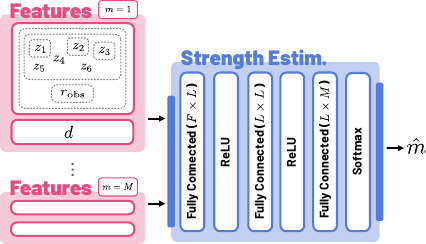}
        \caption{Problem 2a.}
        \label{fig:architecture:2a}
    \end{subfigure}%
    \hspace*{10mm}
    \begin{subfigure}[t]{0.3\textwidth}
        \centering
        \includegraphics[width=\linewidth]{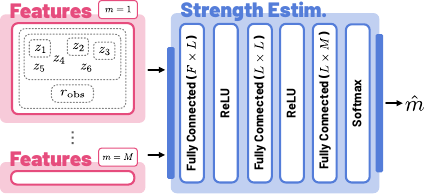}
        \caption{Problem 2b.}
        \label{fig:architecture:2b}
    \end{subfigure}%
    \caption{\textbf{Architectures of the learning-based estimation schemes.} We show the features (red) and architectures (blue) used for distance estimation with known source strength (Problem~\ref{prob:distance:known_strength}) and unknown source strength (Problem~\ref{prob:distance:unknown_strength}) as well as the features and architectures for source strength estimation when the source distance is known (Problem~\ref{prob:strength:known_distance}) and unknown (Problem~\ref{prob:strength:unknown_distance}). Exemplary combinations of features that will be investigated in Section~\ref{sec:eval} are indicated by the dotted gray bounding boxes.}
\end{figure*}

\scaleSubsection\subsection{Learning-based Source Distance Estimation}\scaleSubsectionBelow\label{sec:schemes:learning}
The \ac{LC} scheme introduced in Proposition~1 exploits only the average behavior based on~\eqref{eq:schemes:lc:exp_mol_ratio} and~\eqref{eq:schemes:lc:d}, but neglects other statistical information contained in the observation. Due to the inherent complexity of turbulent flow, which determines the statistical properties of the received signal, following~\cite{rigolli:learning_predict_target_location_turbulent_odor_plumes}, we adopt a data-driven approach to develop a learning-based estimator that is able to extract more information contained in the observations. 

In particular, we propose for Problems~\ref{prob:distance:known_strength} and~\ref{prob:distance:unknown_strength} feed-forward neural networks, as shown in Figures~\ref{fig:architecture:1a} and~\ref{fig:architecture:1b}, with two hidden layers, each with $L=64$ nodes and followed by a $\mathrm{ReLU}(\cdot)$ activation function to estimate $d$ based on a feature vector $\z$ combining $F$ different features computed from $y_1[l]$ and $y_2[l]$.
Besides $r_{\mathrm{obs}}$, we can use the average intensity of $y_1[l]$ during the observation interval, $z_1 = \Ehat{l \cdot T_\mathrm{s} \in \mathcal{T}}{y_1[l]}$, as a feature\footnote{Note that computing $z_1$ for both $y_1[l]$ and $y_2[l]$ contains the same information as computing $z_1$ for $y_1[l]$ in addition to $r_\mathrm{obs}$.} and/or the features proposed in~\cite{rigolli:learning_predict_target_location_turbulent_odor_plumes}, which we detail in the following\footnote{To isolate the contribution of $r_\mathrm{obs}$ compared to that of $z_1$-$z_6$, which have been established for non-degradable molecules in~\cite{rigolli:learning_predict_target_location_turbulent_odor_plumes}, we compute $z_1$-$z_6$ only for the non-degradable molecule type.}:
\begin{description}
   \item[Average Whiff Intensity.] The average intensity of $y_1[l]$ during the whiffs, i.e., 
   \begin{equation}\label{eq:schemes:z2}
        z_2 = \Ehat{l \cdot T_\mathrm{s} \in \mathcal{T}}{y_1[l] \big|y_1[l]\geq y_{\mathrm{whiff},1}}.
    \end{equation}
   \item[Average Whiff Intensity Slope.] The average magnitude of the slope of $y_1[l]$ during whiffs, i.e., 
   \begin{equation}
        z_3 = \Ehat{l \cdot T_\mathrm{s} \in \mathcal{T}}{|\partial_l y_1[l]| \,\big|y_1[l]\geq y_{\mathrm{whiff},1}}.
    \end{equation}
   \item[Average Whiff Duration.] The average duration of a whiff in the observation interval, i.e., 
    \begin{equation}
       z_4 =  \frac{1}{N_\mathrm{W}(\mathcal{T})} \Ehat{l \cdot T_\mathrm{s} \in \mathcal{T}}{\boldsymbol{1}\{y_1[l]\geq y_{\mathrm{whiff},1} \} },
   \end{equation}
   where $N_{\mathrm{W}}(\mathcal{T})$ is the number of whiffs in the observation interval $\mathcal{T}$.
   \item[Average Blank Duration.] The average duration of a blank in the observation interval, i.e., 
   \begin{equation}
       z_5 = \frac{1}{N_\mathrm{B}(\mathcal{T})} \Ehat{l \cdot T_\mathrm{s} \in \mathcal{T}}{\boldsymbol{1}\{y_1[l]< y_{\mathrm{whiff},1} \} },
   \end{equation}
   where $N_{\mathrm{B}}(\mathcal{T})$ is the number of blanks in the observation interval $\mathcal{T}$.
   
   \item[Intermittency Factor.] The fraction of the observation interval where whiffs occur, i.e., 
   \begin{equation}
       z_6 =  \Ehat{l \cdot T_\mathrm{s} \in \mathcal{T}}{\boldsymbol{1}\{y_1[l]\geq y_{\mathrm{whiff},1} \} }.
   \end{equation}
\end{description}
In the remainder of this paper, we will call $z_1$-$z_6$ \ac{SM} features as they can be computed from a single molecule type. 

We train the neural network on pairs of feature vectors and distances $d$ using the mean squared error loss, i.e., the loss $\mathcal{L}(d, \hat{d})$ for a single sample with true distance $d$ and estimated distance $\hat{d}$ is defined as
\begin{equation}
    \mathcal{L}(d, \hat{d}) = \left(d-\hat{d}\right)^2.
\end{equation}
We use the $\mathrm{Adam}$ optimizer with learning rate $10^{-3}$ and weight decay of strength $10^{-4}$. The feature vectors can contain any subset $\mathcal{Z}$ of $\left\{r_\mathrm{obs}, z_1, \dots, z_6 \right\}$. We use a batch size of $B=200$ for at most 1000 epochs and apply early-stopping if the loss does not improve after 10 epochs.
In preliminary experiments, it proved useful to use $\log$-transformed versions of $\{z_1, \dots, z_6\}$, as especially the concentration-based features can vary over two orders of magnitude, which hinders effective learning. Therefore, we used $\log$ features throughout our evaluation.
As illustrated in Figure~\ref{fig:architecture:1a} and~\ref{fig:architecture:1b}, depending on whether the network is trained for Problem~\ref{prob:distance:known_strength} or~\ref{prob:distance:unknown_strength}, the training data contains either only samples for a specific source strength or samples for all source strengths~$\forall J^{(m)} \in \mathcal{J}$.

\section{Source Strength Estimation Schemes}\label{sec:strength}
In this section, we consider the problem of source strength estimation, i.e., estimating the rate with which a source releases molecules. While source strength estimation has also been considered in the \ac{SMC} literature, the existing schemes (e.g.,~\cite{noel:joint_channel_parameter_estimation_diffusive_mc,tiwari:estimate_forward_relaying_diffusion_mc}) were designed for the case where molecules propagate via diffusion and/or laminar flow and are thus not applicable to turbulent flow.
In the following, we first formalize two types of source strength estimation problems before introducing an \ac{LC} strength estimator in Section~\ref{sec:strength:lc} and a learning-based scheme in Section~\ref{sec:strength:learning}. 

\stepcounter{problem}
\begin{subproblem}[Source Strength Estimation with Known Source Distance]\label{prob:strength:known_distance}
We consider an \ac{RX}, as defined in Section~\ref{sec:system_model:rx}, that is positioned downwind of the \ac{TX} at a known distance $d = \lVert \xtx-\xrx \rVert_2$.
The \ac{TX} releases molecules with an unknown flux strength $J^{(m)} \in \mathcal{J}$, where $p_1$, $p_2$, and the set of possible flux strengths $\mathcal{J}$ are known but $m$ is unknown.
The goal of the \ac{RX} is then to estimate $m$ based on the observed numbers of type 1 and type 2 molecules, i.e., $y_1[l]$ and $y_2[l]$, that are sampled during the observation interval $\mathcal{T}$. 
\end{subproblem}

\begin{subproblem}[Source Strength Estimation with Unknown Source Distance]\label{prob:strength:unknown_distance}
This problem is identical to Problem~\ref{prob:strength:known_distance} except that the distance $d$ between \ac{TX} and \ac{RX} is also unknown.
\end{subproblem}

\begin{remark}
    Problem~\ref{prob:strength:unknown_distance} is generally more challenging compared to Problem~\ref{prob:strength:known_distance} because the source strength estimator also has to estimate the source distance to account for the increasing attenuation as the distance between \ac{TX} and \ac{RX} grows.
\end{remark}

\subsection{Low-Complexity Source Strength Estimation}\label{sec:strength:lc}
When estimating the strength of a molecule source, a central challenge is that the signal intensity observed at the \ac{RX} not only depends on the source strength $J$ but also on the displacement vector between \ac{TX} and \ac{RX}, $\boldsymbol{\delta}=\xtx-\xrx$. Therefore, the \ac{RX} must know i) $\boldsymbol{\delta}$ and ii) how $\boldsymbol{\delta}$ affects the signal intensity.

While it is not generally known how $\boldsymbol{\delta}$ affects $y_1[l]$ or $y_2[l]$, insights for some features computable from $y_1[l]$ or $y_2[l]$ have been provided in the literature. 
In particular, it has been shown in~\cite{celani:odor_landscapes_turbulent_environments} that the expected non-zero concentration of the non-degraded molecule type, $c_{1,>0} = \E{y_1[l]|y_1[l]>0}$, increases linearly with $J$ and depends on $\boldsymbol{\delta}$, i.e., it holds that~\cite{celani:odor_landscapes_turbulent_environments}
\begin{equation}\label{eq:strength:whiff_intensity}
    c_{1,>0} \approx \frac{J}{g(\boldsymbol{\delta})},
\end{equation}
where $g(\cdot)$ is a function of the displacement vector between \ac{TX} and \ac{RX}. While $g(\cdot)$ generally depends on the boundaries of the environment and the specific characteristics of the turbulence, it was shown in~\cite{celani:odor_landscapes_turbulent_environments} that $g(\xtx-\xrx)$ generally scales as a power-law of the downwind distance, i.e., 
\begin{equation}\label{eq:strength:g_approx}
    g(\boldsymbol{\delta}) \propto \left|[\boldsymbol{\delta}]_1\right|^\beta,
\end{equation} 
where $\beta>0$ is a constant that depends on the properties of the turbulent flow and $[\boldsymbol{\delta}]_1$ extracts the first element of the displacement vector, i.e., in our case the downwind distance between \ac{TX} and \ac{RX}. From~\eqref{eq:strength:whiff_intensity} and~\eqref{eq:strength:g_approx} it becomes apparent that $c_{1,>0}$ can be used to estimate $J$ for a given $g(\cdot)$. In practice, $c_{1,>0}$ is not directly available at the \ac{RX}. Instead, the average whiff intensity $z_2$ defined in~\eqref{eq:schemes:z2} can be used because $z_2$ is the empirical counterpart of $c_{1,>0}$, where the condition $y_1[l]>0$ is replaced by the whiff condition $y_1[l] \geq y_{\mathrm{whiff},1}$.

\begin{theorem}\label{prop:strength:lc}
    The index $m \in \{1, \dots, M\}$ of the source strength $J^{(m)} \in \mathcal{J}$ can be estimated as follows
    \begin{equation}\label{eq:strength:mhat}
        \hat{m} = \arg\min_{m \in \{1, \dots, M\}} \left| z_2 \cdot d^\beta - \mu_m \right|.
    \end{equation}    
    For Problem~\ref{prob:strength:known_distance}, $d$ is the known distance between \ac{TX} and \ac{RX}. For Problem~\ref{prob:strength:unknown_distance}, $d$ in~\eqref{eq:strength:mhat} has to be replaced with the estimated distance $\hat{d}^\mathrm{LC}$ from~\eqref{eq:schemes:lc:estimator}, i.e., $d=\hat{d}^\mathrm{LC}$. 
    Furthermore, $\mu_m$, $m \in \{1,\dots,M\}$, is defined as 
    \begin{equation}
         \mu_m = \mu_1+(m-1)\cdot\frac{\mu_M-\mu_1}{M-1}.
    \end{equation}
    Here, $\mu_{m'}$, $m' \in \{1,M\},$ denotes the empirical approximation of $\mathrm{E}\{z_2 \cdot d^\beta|J^{(m')}\}$, where the expectation is taken over the possible displacement vectors $\boldsymbol{\delta}$.
\end{theorem}
\begin{proof}
    First, we note that $g(\cdot)$ can be written as in~\eqref{eq:strength:g_approx}. Since the propagation environment is larger in downwind direction than in crosswind direction (see Section~\ref{sec:system_model:channel}), $d=\lVert \xtx-\xrx\rVert_2$ approximates the downwind distance $\left|[\boldsymbol{\delta}]_1\right|$. Thus, we obtain $z_2 \propto \frac{J}{d^\beta}$ and conversely
    \begin{equation}\label{eq:strength:J}
        z_2 \cdot d^\beta \propto J.
    \end{equation}
    Then, $z_2 \cdot d^\beta$ is compared to the expected value for a given source strength, $\mu_m$, $m \in \{1,\dots,M\}$. Because we assume that the $J^{(m)}$ are equally spaced between $J^{(1)}$ and $J^{(M)}$ and exploiting the linearity of $z_2$ in $J$, it becomes apparent that knowing $\mu_1 = \E{z_2 \cdot d^\beta|J^{(1)}}$ and $\mu_M =\E{z_2 \cdot d^\beta|J^{(M)}}$ is sufficient to infer $\mu_m$ $\forall m$, resulting in~\eqref{eq:strength:mhat}. 
    Note that $\mu_1$ and $\mu_M$ implicitly contain the environment-dependent proportionality constant from~\eqref{eq:strength:g_approx} and \eqref{eq:strength:J} and no other constants are needed. 
\end{proof}
\begin{remark}
    The \ac{LC} source strength estimator requires knowledge of $\beta$, $\mu_1$, and $\mu_M$, all of which depend on the propagation environment. 
    In practice, they could be estimated in settings where the source strength is known: For synthetic sources, this could be implemented via pilot symbols. 
    For natural sources, this would require \textit{offline} measurements in the environment prior to the deployment. 
    First, the actual source strength needs to be measured at the source as ground truth. Then, if $z_2$ is measured at several positions,~\eqref{eq:strength:J} can be used to estimate $\beta$. Afterwards, $\mu_1$ and $\mu_M$ can be estimated from the same measurements used to estimate $\beta$.
\end{remark}

\subsection{Learning-based Source Strength Estimation}\label{sec:strength:learning}
Analogously to the learning-based distance estimation scheme reported in Section~\ref{sec:schemes:learning}, we propose a learning-based source strength estimator based on a simple feed-forward neural network. 
The proposed networks for source strength estimation are depicted in Figure~\ref{fig:architecture:2a} and~\ref{fig:architecture:2b}. The same architecture is used for both Problem~\ref{prob:strength:known_distance} and~\ref{prob:strength:unknown_distance}, only the features differ. Namely, the networks have two hidden layers, each with $L=64$ nodes and followed by a $\mathrm{ReLU}$ activation function. 
For Problem~\ref{prob:strength:unknown_distance}, the network receives a subset $\mathcal{Z}$ of $\{r_{\mathrm{obs}}, z_1, \dots, z_6\}$ as input features. In contrast, for Problem~\ref{prob:strength:known_distance}, the network additionally receives the distance $d$ between \ac{TX} and \ac{RX} as side information. 
In both cases, the goal is to determine the strength of the molecule source. As the set of source strengths is discrete, the task of identifying the source strength $J^{(m)} \in \mathcal{J}$ can be framed as a classification problem. Thus, the final layer yields an $M$-dimensional vector and is followed by a $\mathrm{Softmax}$ activation function. 

Typically, training neural networks for classification problems is based on the cross-entropy loss, which treats all errors equally~\cite{oshea:introduction_dl_phy_layer}. For example, the neural network receives the same loss irrespective of whether $|\hat{m}-m|=1$ or $|\hat{m}-m|=4$.
However, for both natural and synthetic molecule sources, not all classification errors are equally problematic when estimating the source strength. 
When sensing a natural molecule source, a smaller value of $|m-\hat{m}|$ is a more useful estimate.
Also, for \ac{SMC} systems confusing two \textit{neighboring} source strengths may be preferable: If a \ac{TX} employs Gray mapping to reduce the \ac{BER} of a \ac{CSK} scheme, a wrongly estimated source strength index $\hat{m}$ causes only a single bit error if $|m-\hat{m}|=1$, whereas it may cause multiple bit errors for larger values of $|m-\hat{m}|$. 
In both cases, the standard cross-entropy loss would ignore the ordinal structure of the source strength estimation problem. 

To alleviate this, we instead propose a custom loss function, which extends the default cross-entropy loss by a term that is larger for increasing $|\hat{m}-m|$. The loss function is conceptually similar to loss functions proposed for ordinal regression (see, e.g.,~\cite{polat:class_distance_weighted_cross_entropy_loss}) and, for a single sample with true source strength index $m$, is given as follows
\begin{equation}\label{eq:strength:learning:loss}
    \mathcal{L}\left(\hat{p}(\hat{m}),m\right) = \mathrm{CrossEntropy}\left(\hat{p}(\hat{m}), m\right) + \lambda \sum_{\hat{m}} \hat{p}(\hat{m}) \cdot  |\hat{m}-m|,
\end{equation}
where $\mathrm{CrossEntropy(\cdot,\cdot)}$ denotes the standard cross-entropy loss~\cite{oshea:introduction_dl_phy_layer}, $\lambda\geq0$ is a weighting factor and $\hat{p}(\hat{m})$ is the $\hat{m}$-th entry in the output vector of the $\mathrm{Softmax}$ layer of the neural network.
In summary, \eqref{eq:strength:learning:loss} essentially implements a \textit{hybrid} between a classification and a regression problem: If $\lambda=0$, the network defaults to a classification problem. If $\lambda>0$, the cross-entropy loss is augmented by the regression loss which punishes the neural network more if it assigns a higher probability to a source strength that is further away from the true source strength. 

Similar to the training of the proposed distance estimation networks, we train the network via the Adam optimizer with learning rate $10^{-3}$ and additional weight decay of strength $10^{-4}$. We use a batch size of $B=200$ for 100 epochs.

\section{Performance Evaluation}\label{sec:eval}
In this section, we validate the proposed schemes using simulations of particle propagation in complex turbulent flow. We first describe the simulation setup in detail in Section~\ref{sec:eval:setup} and introduce the utilized evaluation metrics in Section~\ref{sec:eval:metrics}. In Sections~\ref{sec:eval:distance:known} and~\ref{sec:eval:distance:unknown}, we evaluate the distance estimation schemes respectively for known and unknown source strengths. Finally, we evaluate the strength estimation schemes in Section~\ref{sec:eval:strength}.

\subsection{Simulation Setup}\label{sec:eval:setup}
\subsubsection{Employed Dataset and Particle Propagation}\label{sec:eval:setup:dataset}
To validate the proposed schemes, we leverage the dataset provided in~\cite{biferale:turb_smoke_database_lagrangian_pollutants}. It contains the individual positions of several million Lagrangian tracer particles propagating through complex turbulent flow. It has been specifically developed to enable the evaluation of algorithms designed for tasks such as olfactory search or leakage localization and has been already utilized extensively for such studies (see, e.g., in~\cite{piro:policy_heterogeneity_improves_collective_olfactory_search,piro:many_wrong_models_approach_localise_odor_source,heinonen:exploring_bayesian_olfactory_search}). 
The dataset reports distances and timescales in a normalized form as multiples of $\pi$ and the Kolmogorov timescale $\tau_\eta$, respectively. 
While these normalized quantities are not easily related to real-world distances and timescales, we would like to emphasize that the employed dataset covers all typical effects of turbulence such as intermittency and non-Gaussian concentration statistics. Therefore, the dataset is suitable for evaluating our proposed algorithms which are applicable independent of the actual time and length scale. 

The dataset was generated in a two-step procedure: First, the velocity field of a \textit{Homogeneous Isotropic Turbulent Flow} without mean wind component was computed via direct numerical integration of the incompressible Navier-Stokes equations~\eqref{eq:system:navier_stokes}, \eqref{eq:system:navier_stokes:divergence-free}. In a second step, the propagation of the individual Lagrangian tracer particles in this turbulent flow was simulated via numerical solution of~\eqref{eq:system:dx} with $\uvec(\x,t) = \uvec_\mathrm{turb}(\x,t)$. In the dataset, a continuous molecule release is assumed, i.e., after each Kolmogorov time-scale, $J^{(m)}$ molecules appear in the \ac{TX} volume. The resulting positions of the individual particles are reported after every Kolmogorov time scale $\tau_\eta$ in~\cite{biferale:turb_smoke_database_lagrangian_pollutants}.

Next, we leverage this dataset as follows to generate the received signals used by the proposed schemes: Once a new molecule has been created at the \ac{TX}, we mark the molecule as \textit{active} with probability $p_\mathrm{active} = J^{(m)}/J^{(M)}$ and discard it otherwise. Thus, all molecules are used in the simulations if $m=M$. Each active molecule is then determined to be of type 1 and type 2 with probability $p_1$ and $p_2=1-p_1$, respectively. 
The molecules then propagate through the simulation domain, a cuboid of size $\mathcal{S} = [-0.25\pi, 7\pi] \times [-1.625\pi, 1.625\pi] \times [-1.625\pi, 1.625\pi]$ by adding the mean-wind component $\uvec_0\cdot t$ to the position of the $i$-th molecule $\x_i(l \cdot T_\mathrm{s})$ in timestep $l$. In every time step, each molecule of type 2 is removed from the simulation due to degradation with probability $p_{\mathrm{deg},2}$. 
In the following, we describe how the resulting particle positions are then used to generate the signals observed at the \ac{RX} used by the proposed schemes for strength and distance estimation.

{\renewcommand{\arraystretch}{1.7}
\setlength{\tabcolsep}{2pt}
\begin{table*}
  \centering
  \caption{Default system parameters for the simulations.}\label{tab:params}
  \begin{tabular}{|c|l|r l|}
    \hline
    \rowcolor{gray!20}
    \textbf{Variable} & \textbf{Meaning} & & \textbf{Value} \\ \hline\hline
    \multirow{3}{*}{$p_1$} & \multirow{3}{*}{Fraction of non-degradable molecules} & Baseline Scenario: &$1.0$ \\ \cline{3-4}
    & & Mixture Scenario: &$0.34$ \\ \cline{3-4}
    & & Fig.~\ref{fig:eval:release-ratio}: & $\{0.1, \dots, 0.9\}$ \\ \hline
    \multirow{3}{*}{$p_2$} & \multirow{3}{*}{Fraction of degradable molecules} & Baseline Scenario: &$0.0$ \\ \cline{3-4}
    & & Mixture Scenario: & $0.66$ \\ \cline{3-4}
    & & Fig.~\ref{fig:eval:release-ratio}: &$\{0.1, \dots, 0.9\}$ \\ \hline
    \multirow{3}{*}{$p_{\mathrm{deg},2}$} & \multirow{3}{*}{Degradation probability} & Figs~\ref{fig:eval:feature-distance-prediction}, \ref{fig:eval:feature_combinations}, \ref{fig:eval:distance-dependent-errors}-\ref{fig:eval:strength:weighing-factor}: &$0.03$ \\ \cline{3-4}
    & & Fig.~\ref{fig:eval:deg_rates}: & $\{0.01, \dots, 0.09\}$ \\ \cline{3-4}
    & & Fig.~\ref{fig:eval:release-ratio}: & $\{0.03, 0.07\}$ \\ \hline
    \multirow{2}{*}{$p_\mathrm{active}$} & \multirow{2}{*}{Fraction of active molecules $p_\mathrm{active} = J^{(m)}/J^{(M)}$} & Figs.~\ref{fig:eval:feature-distance-prediction}-\ref{fig:eval:distance:interval-length}: &$1.0$ \\ \cline{3-4}
    & & Figs.~\ref{fig:eval:source_strength_dependent_distance_estimation}-\ref{fig:eval:strength:weighing-factor}: &$\{0.1, 0.2, \dots, 1.0\}$ \\ \hline
    $J^{(M)}$ & Maximum number of released particles per $\tau_\eta$ & & 3000 \\ \hline\hline

    $\mathcal{S}$ & Simulation domain & \multicolumn{2}{|c|}{$[-0.25\pi, 7\pi] \times [-1.625\pi, 1.625\pi] \times [-1.625\pi, 1.625\pi]$} \\ \hline
    & Simulation duration & & $900 \tau_\eta$ \\ \hline\hline

    $u/u_\mathrm{rms}$ & \makecell[l]{Ratio of mean wind $u$ and \\ root mean squared value of the velocity field} & & 2.05 \\ \hline\hline 
    
    $\xtx$ & \ac{TX} position & \multicolumn{2}{|c|}{ $[0,\, 0,\, 0]^{\transpose}$} \\ \hline 
    $\rtx$ & \ac{TX} radius & & $\frac{3}{256}\pi$ \\ \hline 
    $\xrx$ & \ac{RX} position & \multicolumn{2}{|c|}{ $\mathcal{U}\left([0, 7\pi] \times [-1.625\pi, 1.625\pi] \times [-1.625\pi, 1.625\pi]\right)$} \\ \hline 
    $\rrx$ & \ac{RX} radius & & $0.1\pi$ \\ \hline\hline 

    $N_\mathrm{RX}$ & Number of considered \ac{RX} positions & & 1,000 \\ \hline 
    \multirow{4}{*}{$T$} & \multirow{4}{*}{Observation interval length} &  Figs.~\ref{fig:eval:feature-distance-prediction}-\ref{fig:eval:distance-dependent-errors},\ref{fig:eval:source_strength_dependent_distance_estimation}:& $100 \tau_\eta$ \\ \cline{3-4} 
    & & Figs.~\ref{fig:eval:strength:main}-\ref{fig:eval:strength:weighing-factor}: & $500 \tau_\eta$ \\ \cline{3-4} 
    & & Fig.~\ref{fig:eval:distance:interval-length}: & $\{10 \tau_\eta, \dots, 100 \tau_\eta\}$ \\ \cline{3-4} 
    & & Fig.~\ref{fig:eval:strength:observation_length} & $\{10 \tau_\eta, 100 \tau_\eta, 500 \tau_\eta\}$ \\ \hline
    \multirow{2}{*}{$O$} & \multirow{2}{*}{\# Observation intervals per source strength} & Figs.~\ref{fig:eval:feature-distance-prediction}-\ref{fig:eval:source_strength_dependent_distance_estimation}: & 10,000 \\ \cline{3-4} 
    & & Figs~\ref{fig:eval:strength:main}-\ref{fig:eval:strength:weighing-factor}: & 1,000 \\ \hline
  \end{tabular}
\end{table*}
\renewcommand{\arraystretch}{1.0}}

\subsubsection{Evaluation Procedure}\label{sec:eval:setup:simulation}
To evaluate the proposed schemes, as baseline, we use learning-based estimation schemes leveraging only $z_1$-$z_6$ (see Section~\ref{sec:schemes:learning} and Section~\ref{sec:strength:learning}). To endow the baseline schemes with the same molecule budget as the multi-molecule schemes, we perform two types of simulations:
\begin{itemize}
    \item \textbf{Baseline Scenario:} The \ac{TX} releases only non-degradable molecule types, i.e., $p_1=1$. 
    \item \textbf{Mixture Scenario:} The \ac{TX} releases the non-degradable molecule type with probability $p_1$ and the degradable molecule type with probability $p_2=1-p_1$. 
\end{itemize}
Since $z_1$-$z_6$ are computed only using molecules of type 1, the baseline schemes have a slight advantage over the proposed multi-molecule-type schemes due to reduced counting noise. 

To generate the training data for the learning-based schemes and the test data used to evaluate all considered schemes, we use the following procedure: 
We consider $N_\mathrm{RX}=1000$ random positions for the \ac{RX} in the simulation domain and select for each source strength initially $O$ observation intervals of length $|\mathcal{T}|$ with random starting times after $t=100\tau_\eta$. The earliest possible starting time, $t=100\tau_\eta$, is chosen to avoid the transient phase where molecules have not yet propagated through the whole simulation domain.

To ensure a fair comparison, the observation intervals for the baseline scenario and the mixture scenario are matched: for each interval in the baseline scenario, an interval with the same starting time and the same \ac{RX} position is generated in the mixture scenario. If no molecules are observed in an interval, it cannot be used to infer the source distance, and is therefore discarded. 
However, we apply this procedure based on the baseline scenario only: an interval is discarded from both scenarios if no molecules are observed in the baseline scenario, whereas intervals that are empty only in the mixture scenario are retained.
Thus, the baseline schemes relying on a single molecule type have again a slight advantage. 
The remaining \textit{valid} intervals are then used to compute the features $z_1$-$z_6$ and, for the mixture scenario, $r_\mathrm{obs}$. In the few cases in which a feature cannot be computed even if the observation interval contains at least one molecule, a $0$ is assigned as default value. Finally, the resulting dataset is split into a training and a test set with a ratio of $3$:$1$. 
The most important parameters used for the simulation and evaluation are summarized in Table~\ref{tab:params}.

\begin{figure*}[ht!]
    \centering
    \includegraphics[width=\linewidth]{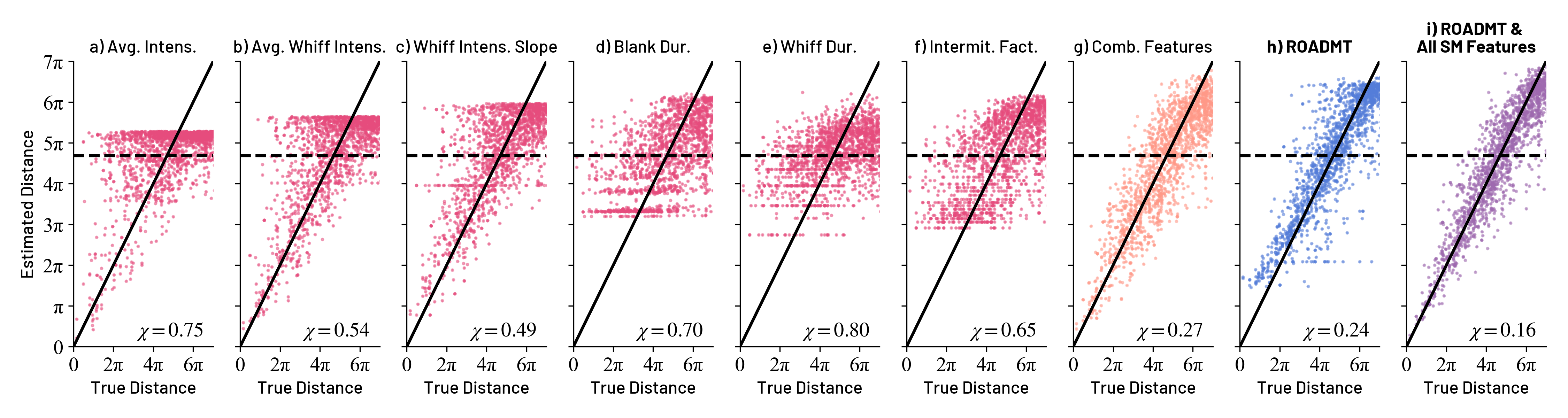}
    \caption{\textbf{Relationship between the true and estimated distance.} We show for each test sample the estimated distance as a function of the true distance (dots) for various features and feature combinations. The solid black line indicates a perfect match between estimated and true distance while the dashed black line shows the mean distance in the test set.}
    \label{fig:eval:feature-distance-prediction}
\end{figure*}

\subsubsection{Parameters of the Estimation Schemes} In the following, we briefly outline how the parameters of the various proposed schemes are obtained. 
All learning-based estimation schemes use the training data generated as described above to learn their parameters using the procedures from Sections~\ref{sec:schemes:learning} and~\ref{sec:strength:learning}.

The \ac{LC} distance estimator requires knowledge of the ratio of released molecules $r_0=p_2/p_1$, the degradation probability $p_{\mathrm{deg},2}$, and the parameter related to the turbulence field $v$. 
The ratio of the released molecules $r_0 = p_2/p_1$ and the degradation probability $p_{\mathrm{deg},2}$ are assumed to be prior knowledge about the sources. The parameter $v$ is taken from Figure~\ref{fig:schemes:distance_travel_time} and has a value of $v=0.22$.

The \ac{LC} strength estimator requires knowledge of the scaling parameter $\beta$ and the expected distance-normalized values of $z_2$, $\mu_1$, and $\mu_M$. 
We found that $\beta=2$ enables a good fit between the normalized version of $z_2$ and $J$ for the training data. Similarly, we computed $\mu_1$ and $\mu_M$ as empirical estimates of $\E{z_2 \cdot d^\beta|J^{(1)}}$ and $\E{z_2 \cdot d^\beta|J^{(M)}}$ from the training data. 

\subsection{Evaluation Metrics}\label{sec:eval:metrics}
\subsubsection{Distance Estimation}\label{sec:eval:metrics:distance}
To quantify the performance of the different distance estimation schemes proposed in Section~\ref{sec:schemes}, we utilize the error metric proposed in~\cite{rigolli:learning_predict_target_location_turbulent_odor_plumes}. There, the authors proposed to use the average squared error over all test samples, normalized by the error when guessing always the mean distance $\bar{d}$ of the test set, i.e., 
\begin{equation}
    \chi = \frac{ \sum_{i} (\hat{d}_i - d_i)^2 }{\sum_{i} (d_i - \bar{d})^2},
\end{equation}
where $\hat{d}_i$ and $d_i$ denote respectively the estimated and the true source distance of the $i$-th sample. 
In other words, a trivial estimator guessing always $\bar{d}$ achieves $\chi=1$, while a perfect estimator achieves $\chi=0$.

\subsubsection{Strength Estimation}
To quantify the performance of the strength estimation schemes proposed in Section~\ref{sec:strength}, we utilize the \textit{\ac{CS}} (see, e.g.,~\cite{geng:automatic_age_estimation_based_facial_aging_patterns}). The \ac{CS} $C_k$ describes the fraction of test samples whose error $m-\hat{m}$ between true strength index $m$ and estimated source strength index $\hat{m}$ is at most $k$. Formally, $C_k$ is defined as follows 
\begin{equation}
    C_k = \sum_{i} \mathbf{1}\left\{ |\hat{m}_i-m_i| \leq k \right\} / \sum_i 1,
\end{equation}
where $i$ denotes the $i$-th test sample and $m_i$ is the corresponding index of the source strength, while $\hat{m}_i$ is the estimated index of the source strength. 
Clearly, $C_0$ recovers the classification accuracy, in the following named \textit{top-1 accuracy}.

\subsection{Distance Estimation with Known Source Strength}\label{sec:eval:distance:known}
In this section, we evaluate the proposed distance estimation schemes for Problem~\ref{prob:distance:known_strength}, i.e., when the source strength is known (see Section~\ref{sec:schemes}). 
\subsubsection{Relationship between true and estimated distance for different features}
In Figure~\ref{fig:eval:feature-distance-prediction}, we first consider the learning-based distance estimators from Section~\ref{sec:schemes:learning} to evaluate which features are most informative regarding the source distance. In Figure~\ref{fig:eval:feature-distance-prediction}-a)-f), we show the estimated distance as a function of the true distance for the baseline scenario. 
While all baseline features enable $\chi<1$, the estimated distances are usually close to the mean distance $\bar{d}$ (dashed black line). Also, the estimates lie, depending on the feature, only in a subset of the full range $[0, 7\pi]$, indicating that not every feature is informative for every distance. This also confirms the findings from~\cite{rigolli:learning_predict_target_location_turbulent_odor_plumes} that intensity-based features (Fig.~\ref{fig:eval:feature-distance-prediction}-a)-c)) are most informative for shorter distances while timing-based features (Fig.~\ref{fig:eval:feature-distance-prediction}-d)-f)) are more informative for larger distances. 
On the other hand, when combining all \ac{SM} features $z_1$-$z_6$, as in Figure~\ref{fig:eval:feature-distance-prediction}-g), the estimates become significantly more accurate compared to relying only on individual features.

Next, we consider the molecule-mixture scenario, i.e., the estimates based on $r_{\mathrm{obs}}$ from~\eqref{eq:schemes:lc:robs}, shown in Figure~\ref{fig:eval:feature-distance-prediction}-h). Clearly, in this case, the estimated distances are strongly correlated with the true distances, resulting in an overall test error of $\chi=0.24$, which is even lower than that for the combination of all features in the baseline scenario (see Fig.~\ref{fig:eval:feature-distance-prediction}-g)), highlighting that the \ac{ROAMT} is highly informative of the source distance. 
Finally, Figure~\ref{fig:eval:feature-distance-prediction}-i) shows the estimates when combining $r_{\mathrm{obs}}$ with $z_1$-$z_6$ in the molecule mixture-scenario. Combining all features results in even lower error rates, demonstrating that $r_{\mathrm{obs}}$ can be effectively augmented by other features obtained from the observed signal of non-degradable molecules, $y_1[l]$, that may carry information about the source distance.

\subsubsection{Combination of Different Features}
\begin{figure}
    \centering
    \includegraphics[width=\linewidth]{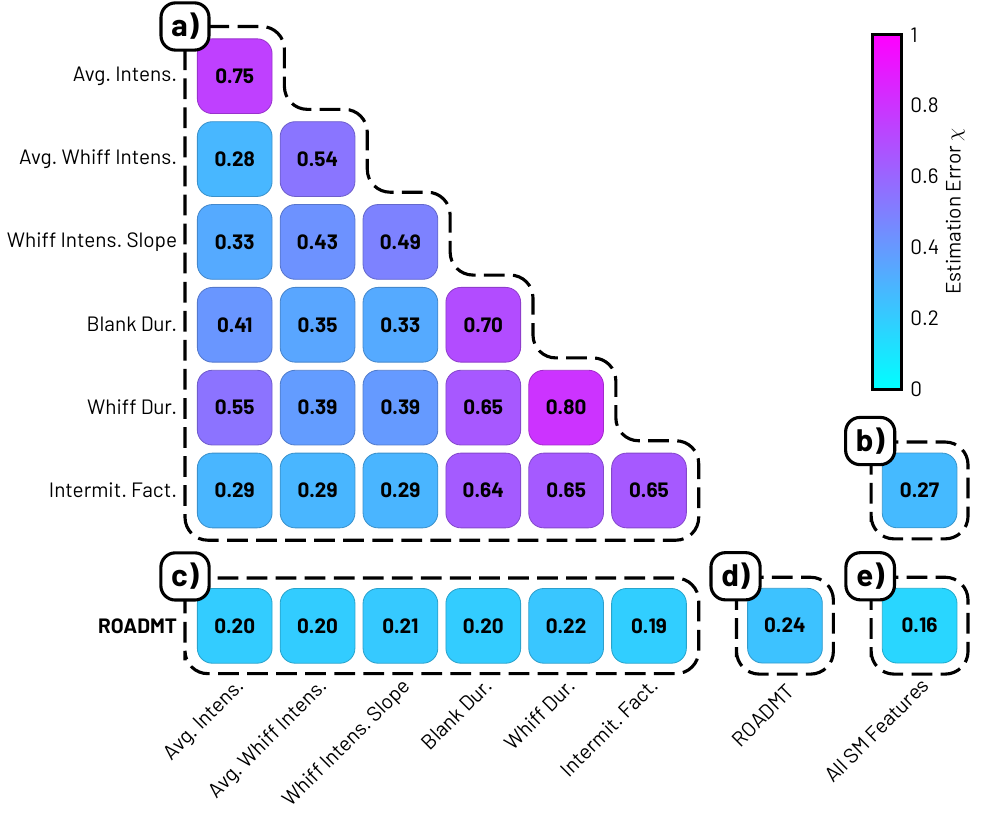}
    \caption{\textbf{Distance estimation error $\chi$ for different feature combinations.} a) Combination of individual features for the baseline scenario. b) Combination of all features for the baseline scenario. c) Combination of $r_{\mathrm{obs}}$ with one of the features $z_1$-$z_6$ for the molecule mixture scenario. d) Only $r_{\mathrm{obs}}$. e) Combination of $r_{\mathrm{obs}}$ with all features~$z_1$-$z_6$.} 
    \label{fig:eval:feature_combinations}
\end{figure}
While Figure~\ref{fig:eval:feature-distance-prediction} revealed the predictive value of individual features and the estimation accuracy when combining all features, Figure~\ref{fig:eval:feature_combinations} investigates \textit{how} to combine individual features in a meaningful manner. 
Figure~\ref{fig:eval:feature_combinations}-a) shows the estimation error when combining two features in the baseline scenario. The observations confirm the results from~\cite{rigolli:learning_predict_target_location_turbulent_odor_plumes}\footnote{Note that~\cite{rigolli:learning_predict_target_location_turbulent_odor_plumes} employed a different dataset for their analysis.}, namely, that combining intensity- (e.g., avg. whiff intensity) and timing-based (e.g., intermittency factor) features minimizes the estimation error. 
Interestingly, combining the average whiff intensity and the average intensity, a feature that was not considered in~\cite{rigolli:learning_predict_target_location_turbulent_odor_plumes}, achieves a low estimation error as well. This is due to the fact that the difference between $z_1$ and $z_2$ implicitly contains information about the intermittency factor. 
Figure~\ref{fig:eval:feature_combinations}-b) shows the estimation error when combining all features in the baseline scenario, as shown in more detail in Figure~\ref{fig:eval:feature-distance-prediction}-g). As also observed in~\cite{rigolli:learning_predict_target_location_turbulent_odor_plumes}, combining more than two features in $\{z_1, \dots, z_6\}$ does not result in a performance gain. 
Figure~\ref{fig:eval:feature_combinations}-c) shows the estimation error when combining $r_\mathrm{obs}$ with one of the features $z_1$-$z_6$ in the molecule mixture scenario. Interestingly, any combination results in a lower estimation error compared to the baseline scenario and compared to relying only on $r_\mathrm{obs}$ (see Figure~\ref{fig:eval:feature_combinations}-d)), confirming that $r_\mathrm{obs}$ introduces new information into the system. 
Finally, Figure~\ref{fig:eval:feature_combinations}-e) shows that combining all \ac{SM} features $z_1$-$z_6$ with $r_\mathrm{obs}$ leads to small improvements compared to combining only individual features $z_1, \dots,z_6$ with $r_\mathrm{obs}$.

\subsubsection{Estimator Comparison for Different Degradation Rates}
\begin{figure}
    \centering
    \includegraphics[width=\linewidth]{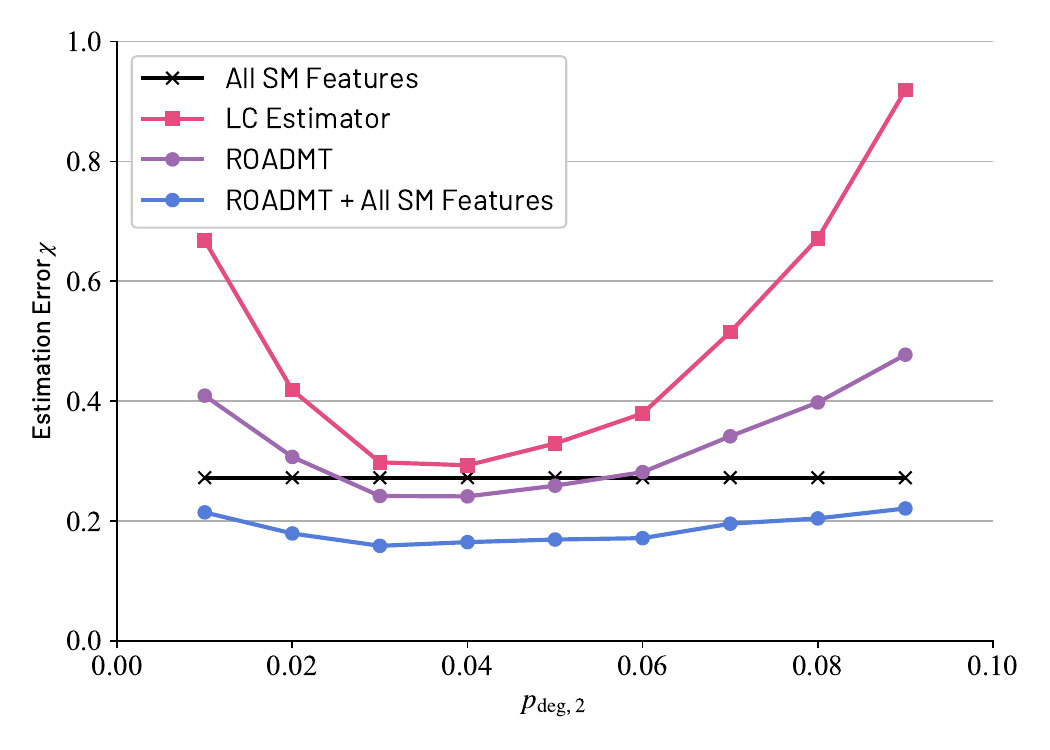}
    \caption{\textbf{Distance estimation error for different degradation rates.} We show the estimation error as a function of degradation probability $p_{\mathrm{deg},2}$ for the \ac{LC} estimator (red), the learning-based scheme relying on $r_\mathrm{obs}$ only (purple), and the learning-based scheme combining $r_\mathrm{obs}$ and $z_1$-$z_6$ (blue). As baseline (black), we show the estimation error when combining $z_1$-$z_6$ for the baseline scenario.}
    \label{fig:eval:deg_rates}
\end{figure}
In Figure~\ref{fig:eval:deg_rates}, we compare the proposed estimation schemes for various degradation probabilities $p_{\mathrm{deg},2}$ to the learning-based estimator relying solely on $z_1$-$z_6$ as baseline (black).
All three proposed schemes share similar trends: As $p_{\mathrm{deg},2} \rightarrow 0$, their estimation errors increase. On the other hand, if $p_{\mathrm{deg},2}$ becomes too large, their errors increase as well. This is due to the fact that for very high values of $p_{\mathrm{deg},2}$ almost no type-$2$ molecules arrive at the \ac{RX}.
While the proposed \ac{LC} scheme (red) has in general an inferior performance compared to the learning-based schemes, Figure~\ref{fig:eval:deg_rates} suggests that, for a suitable choice of the degradation rate, it can closely approach the performance of the single-molecule benchmark that requires extensive training data. 
The learning-based estimator relying only on $r_\mathrm{obs}$ (purple) can, for a suitably chosen degradation rate, outperform the baseline estimator, confirming that $r_\mathrm{obs}$ is a highly informative feature regarding the source distance. 
Finally, combining \ac{ROAMT} and all \ac{SM} features (blue) yields consistently the lowest estimation error over a wide range of degradation rates, demonstrating that the combination of $r_\mathrm{obs}$ with $z_1$-$z_6$ results in a highly robust distance estimation scheme.

\subsubsection{Impact of the Ratio of Released Molecules}
\begin{figure}
    \centering
    \includegraphics[width=\linewidth]{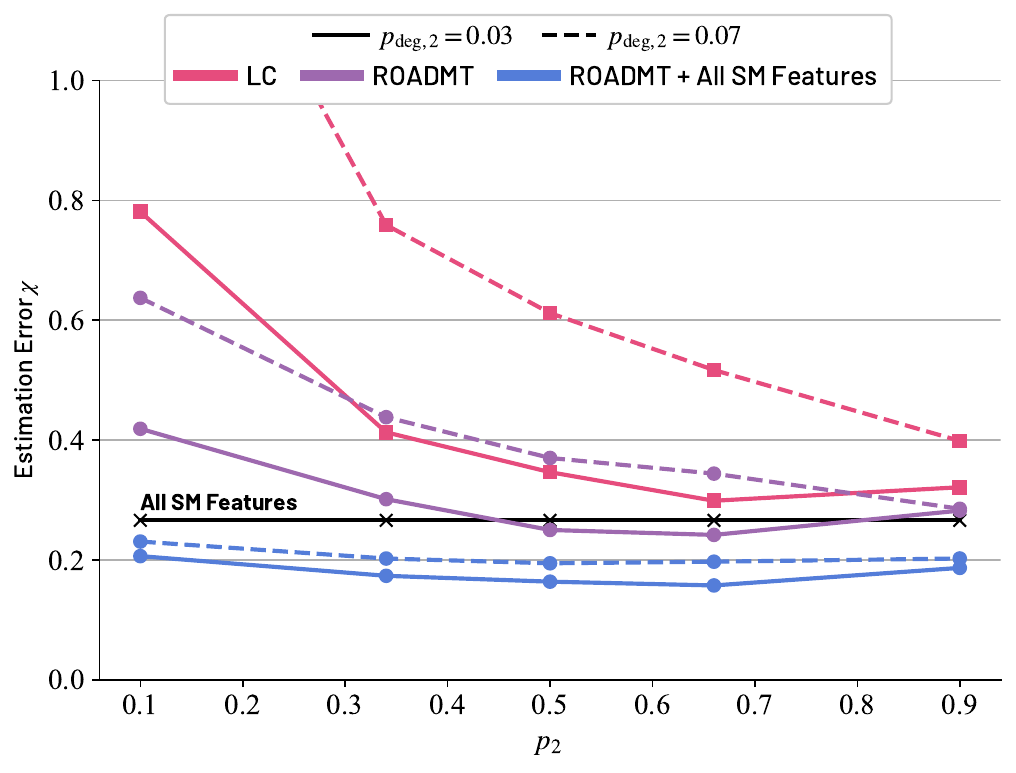}
    \caption{\textbf{Distance estimation error for different release ratios.} We show the estimation error $\chi$ for two different degradation rates as a function of the fraction of degradable molecules. For reference, we show the estimation error of the baseline scenario.}
    \label{fig:eval:release-ratio}
\end{figure}
Next, we investigate how the ratio of the released molecules impacts the estimation performance. To this end, we show in Figure~\ref{fig:eval:release-ratio} the estimation error $\chi$ as a function of $p_2$ for two degradation rates and compare the proposed estimators to the learning-based estimator leveraging only a single molecule type (black).
As before, the \ac{LC} estimator has, for a given $p_{\mathrm{deg},2}$ and $p_2$, always the highest estimation error. However, depending on $p_{\mathrm{deg},2}$, the minimum error may be achieved for different ratios: In particular, for $p_{\mathrm{deg},2}=0.07$, increasing $p_2$ improves performance throughout the considered interval $p_2 \in [0.1, 0.9]$. On the other hand, for $p_{\mathrm{deg},2}=0.03$ a minimum exists at $p_2=0.66$. 
We attribute these different behaviors to the fact that fewer type 2 molecules reach the \ac{RX} if $p_{\mathrm{deg},2}$ is higher for a given release rate. Thus, it may be beneficial to release more of these molecules to reduce the counting noise at the \ac{RX} and consequently enable a more accurate estimation. 
The learning-based estimators leveraging \ac{ROAMT} mirror the overall trends of the respective \ac{LC} estimators, although with overall lower estimation errors. 
Finally, the learning-based estimation error leveraging both \ac{ROAMT} and $z_1$-$z_6$ is least sensitive to the ratio of released molecules. This is a consequence of the relative uncertainty in $z_1$-$z_6$ and $r_\mathrm{obs}$: For smaller values of $p_2$, $z_1$-$z_6$ can be computed with less counting noise. On the other hand, $r_\mathrm{obs}$ may be estimated more accurately for higher values of $p_2$. Therefore, for any value of $p_2$, at least one feature type can be estimated accurately and the other feature type can provide additional information, rendering this scheme most robust to the choice of $p_2$.

\subsubsection{Distance Dependent Errors}
\begin{figure}
    \centering
    \includegraphics[width=\linewidth]{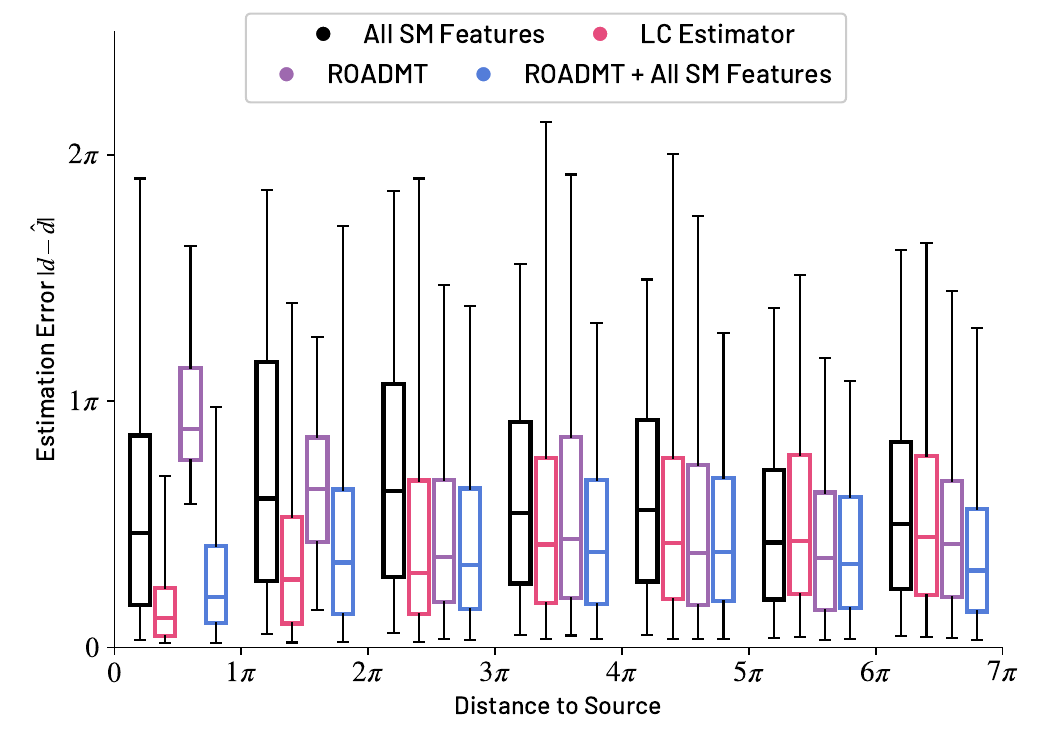}
    \caption{\textbf{Distance-dependent estimation errors.} We show for the various schemes the median estimation error as a function of the true source distance. The bottom and top whisker indicate respectively the $5\%$ and $95\%$ percentiles and the middle line the median of the estimation error.}
    \label{fig:eval:distance-dependent-errors}
\end{figure}
In Figure~\ref{fig:eval:distance-dependent-errors}, we investigate the estimation error for different source distances. To this end, we divide the test samples based on their true source distance $d$ into bins of width $\pi$ and show the distribution of the estimation errors. 
Overall, Figure~\ref{fig:eval:distance-dependent-errors} reveals that most schemes achieve comparable absolute estimation errors $|\hat{d}-d|$ over all distances. 
However, we would like to point out two exceptions: For short distances, the \ac{LC} estimator achieves extremely low errors while the learning-based estimator leveraging only \ac{ROAMT} has high errors. We attribute the \ac{LC} estimator's good performance to the fact that at short distances, much higher molecule counts can be observed (see Figure~\ref{fig:system_model}), thus enabling a very precise estimation of $r_\mathrm{obs}$. At the same time, the \ac{LC} estimator is not negatively affected by the small amount of training samples available for small distances.
On the other hand, already Figure~\ref{fig:eval:feature-distance-prediction}-h) revealed that the learning-based estimator leveraging only \ac{ROAMT} systematically overpredicts short distances, most likely in part due to the small amount of available training samples. Finally, Figure~\ref{fig:eval:distance-dependent-errors} also reveals that the \ac{LC} estimator actually has similar median errors as other schemes. However, it occasionally makes larger errors compared to the other schemes, as can be seen, e.g., by the whiskers for errors between $3\pi$ and $5\pi$. Since $\chi$ is a quadratic error metric, these outliers increase $\chi$ drastically despite a good median performance. 

\subsubsection{Impact of the Observation Interval Length}
\begin{figure}
    \centering
    \includegraphics[width=\linewidth]{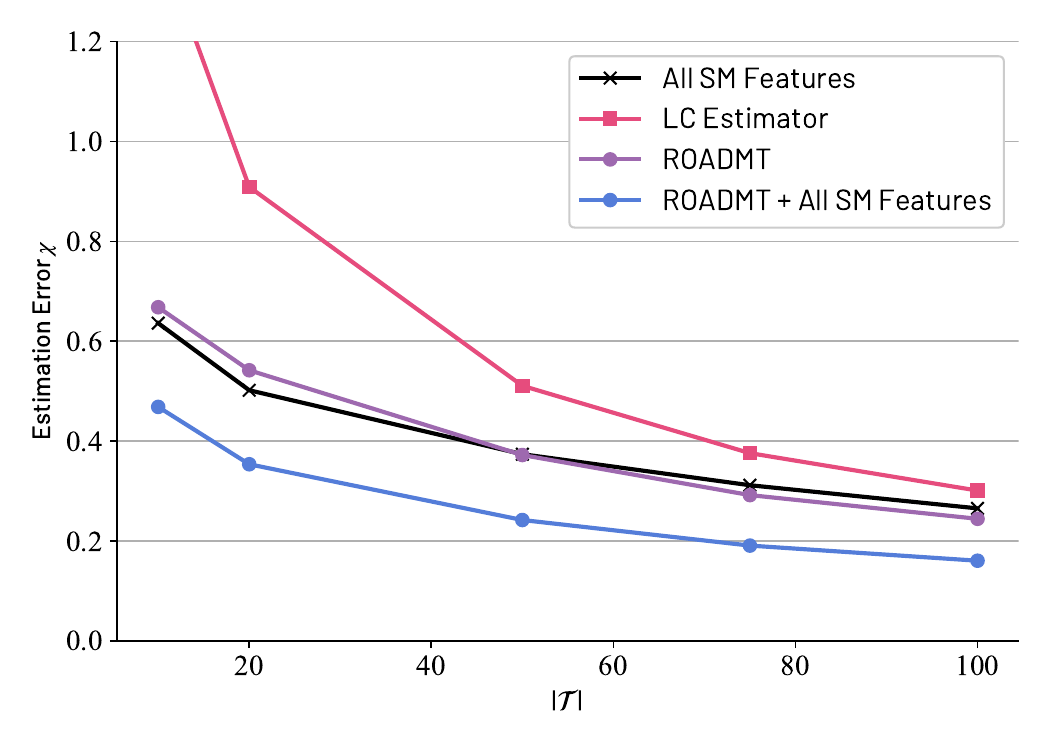}
    \caption{\textbf{Distance estimation error for different observation windows.} We show the distance estimation error $\chi$ as a function of the observation window length $|\mathcal{T}|$ for the proposed \ac{LC} estimator (red), the proposed learning-based estimator leveraging $r_\mathrm{obs}$ (purple), the proposed learning-based estimator leveraging $r_\mathrm{obs}$ and $z_1$-$z_6$ (blue), and the baseline scheme leveraging only \mbox{$z_1$-$z_6$} (black).}
    \label{fig:eval:distance:interval-length}
\end{figure}
Finally, we analyze in Figure~\ref{fig:eval:distance:interval-length} the impact of the observation interval's length on the distance estimation performance. As expected, for all considered estimation schemes, the estimation error $\chi$ decreases as the duration of the observation interval increases. The \ac{LC} estimator benefits most from longer observation intervals: As $|\mathcal{T}|$ increases, the performance gap to the learning-based schemes decreases continuously and approaches almost the baseline estimator's performance for $|\mathcal{T}|=100$. For all values of $|\mathcal{T}|$, the learning-based estimator combining \ac{ROAMT} and $z_1$-$z_6$ achieves the best performance, demonstrating that \ac{ROAMT} can help to achieve a given estimation quality faster compared to schemes relying only on a single molecule type. 

\textbf{Summary}: In this section, we numerically demonstrated that \ac{ROAMT} is an informative feature for source distance estimation. Learning-based distance estimators can combine \ac{ROAMT} with features already established in the literature to achieve an even higher accuracy compared to when they use only a single molecule type. 
Additionally, we showed that an \ac{LC} estimator leveraging \ac{ROAMT} can reliably estimate the source distance without requiring extensive training data. 

\subsection{Distance Estimation with Unknown Source Strength}\label{sec:eval:distance:unknown}
\begin{figure}
    \centering
    \includegraphics[width=\linewidth]{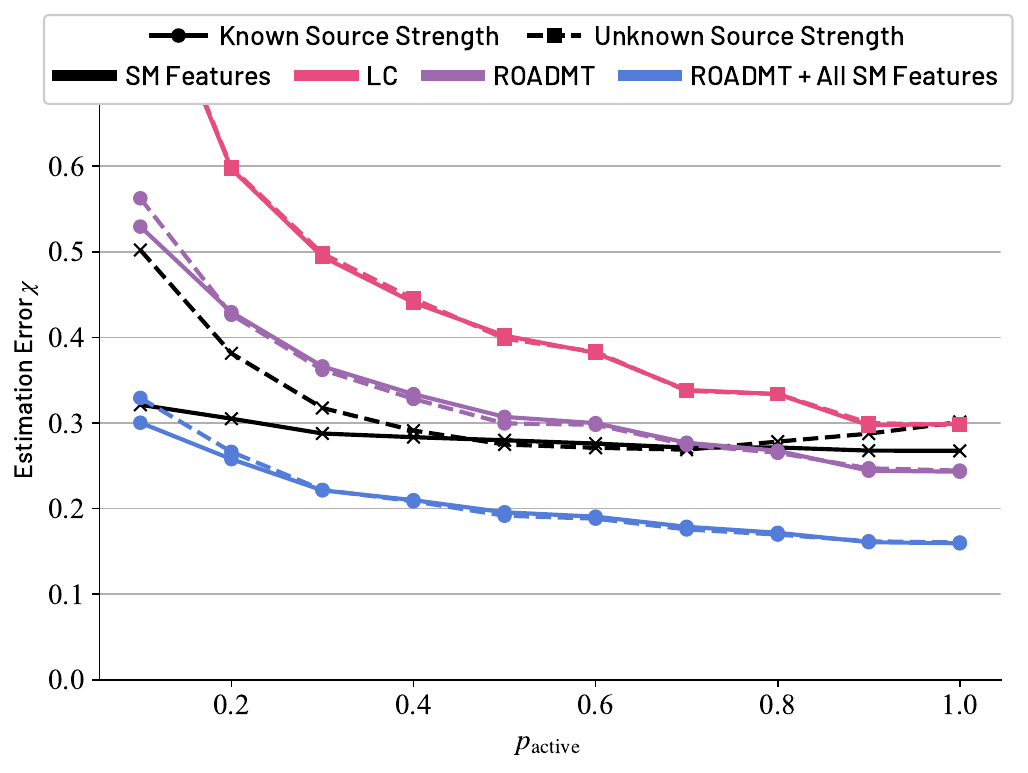}
    \caption{\textbf{Source strength-dependent distance estimation.} We show the estimation error as a function of the source strength, parametrized by $p_\mathrm{active}$, for the \ac{LC} estimator (red), the learning-based scheme relying on $r_\mathrm{obs}$ only (purple), and the learning-based scheme combining $r_\mathrm{obs}$ and $z_1$-$z_6$ (blue). As baseline we consider a learning-based estimator utilizing only \ac{SM} features (black). We show the estimation error for both Problem~\ref{prob:distance:known_strength} (solid) and Problem~\ref{prob:distance:unknown_strength} (dashed).}
    \label{fig:eval:source_strength_dependent_distance_estimation}
\end{figure}
In the previous section, we investigated the performance of the proposed distance estimation schemes when the source strength $J$ was known. Now, we relax this assumption, which may not be fulfilled in many practical scenarios, and consider source distance estimation with unknown source strength, i.e., Problem~\ref{prob:distance:unknown_strength} in Section~\ref{sec:schemes}. 
In Figure~\ref{fig:eval:source_strength_dependent_distance_estimation}, we show the resulting estimation errors $\chi$ as a function of the source strength, determined by $p_\mathrm{active}=J^{(m)}/J^{(M)}$. 
The solid lines indicate schemes that were trained only on data for the specific value of $p_\mathrm{active}$, i.e., Problem~\ref{prob:distance:known_strength}, while the dashed lines indicate schemes that were trained on data containing all possible source strengths $p_\mathrm{active} \in \{0.1, 0.2, \dots, 1.0\}$, i.e., Problem~\ref{prob:distance:unknown_strength}.
First, we note that, as for the case with known source strength, the \ac{LC} estimator has the highest estimation error while the learning-based estimator utilizing \ac{ROAMT} and $z_1$-$z_6$ has the lowest estimation error for a given true source strength. Notably, this scheme outperforms, even when trained for Problem~\ref{prob:distance:unknown_strength}, the baseline scheme trained for the simpler Problem~\ref{prob:distance:known_strength}. 
Similar to our observations from Section~\ref{sec:eval:distance:known}, a higher source strength results in reduced counting noise and, consequently, lower errors, especially for the schemes utilizing \ac{ROAMT}. Furthermore, the \ac{ROAMT}-based schemes have essentially identical performance for Problems~\ref{prob:distance:known_strength} and~\ref{prob:distance:unknown_strength}. In contrast, for Problem~\ref{prob:distance:unknown_strength}, the baseline scheme experiences performance degradation for the extremal values of $p_\mathrm{active}$. 

\textbf{Summary}: In this section, we demonstrated that the schemes utilizing molecule degradation diversity are robust to scenarios with unknown source strength and can consistently achieve a higher estimation accuracy, compared to schemes relying on a single molecule type.

\subsection{Source Strength Estimation}\label{sec:eval:strength}
In this section, we investigate the performance of the proposed source strength estimation schemes from Section~\ref{sec:strength} with a focus on the impact of an unknown distance between \ac{TX} and \ac{RX}. 

\subsubsection{Comparison of the Different Strength Estimation Schemes}
\begin{figure}
    \centering
    \includegraphics[width=\linewidth]{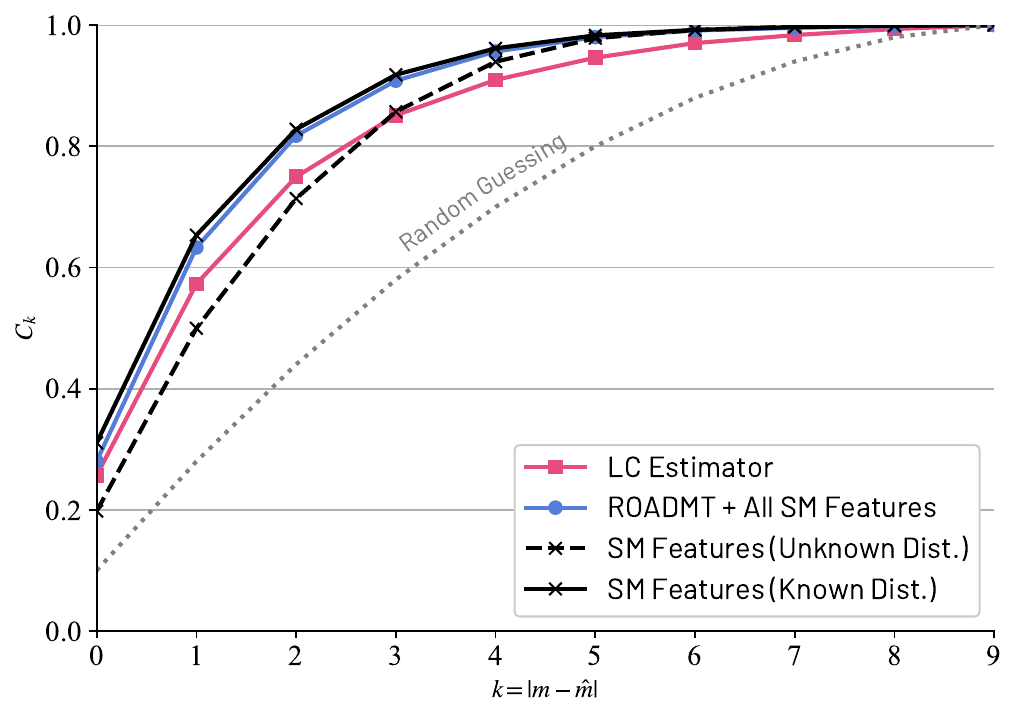}
    \caption{\textbf{Strength estimation performance of different schemes.} We show the \ac{CS} of the proposed \ac{LC} estimator (red) and the proposed learning-based estimator utilizing $r_\mathrm{obs}$ and $z_1$-$z_6$ (blue). For reference, we show the baseline schemes leveraging the \ac{SM} features $z_1$-$z_6$ both when the source distance is known (solid black) and unknown (dashed black). Finally, we show the expected \ac{CS} when randomly guessing the source strength.}
    \label{fig:eval:strength:main}
\end{figure}
In Figure~\ref{fig:eval:strength:main}, we evaluate the performance of the proposed strength estimation schemes. In particular, we show the \ac{CS} of the \ac{LC} strength estimator (red) and the learning-based strength estimator leveraging $r_\mathrm{obs}$ and $z_1$-$z_6$ (blue) for Problem~\ref{prob:strength:unknown_distance}, i.e., when the source distance is unknown. 
For comparison, we also show the results for the baseline scheme utilizing $z_1$-$z_6$ for both Problem~\ref{prob:strength:known_distance} (solid black) and Problem~\ref{prob:strength:unknown_distance} (dashed black). 
The baseline scheme for Problem~\ref{prob:strength:known_distance} represents a genie-aided reference that is not subject to the ambiguity introduced by the unknown distance to the source. 

First, we observe that all schemes, including the baseline schemes, perform significantly better than randomly guessing the source strength (dotted grey). As expected, the baseline scheme with known source distance experiences a significant performance gain compared to the baseline scheme with unknown source distance.
Notably, the proposed learning-based scheme follows the \ac{CS} of the baseline scheme with known source distance closely, indicating that \ac{ROAMT} can help to effectively mitigate the ambiguity between source distance and source strength.
Finally, the proposed \ac{LC} estimator achieves a higher \ac{CS} for $k=0$ compared to the baseline scheme with unknown source distance. However, the \ac{CS} of the proposed \ac{LC} estimator reaches 1 only for a larger $k$ compared to all other schemes, indicating that the \ac{LC} strength estimator misjudges the source strength drastically in some cases.

\subsubsection{Comparison of Different Features for Source Strength Estimation}
\begin{figure}
    \centering
    \includegraphics[width=\linewidth]{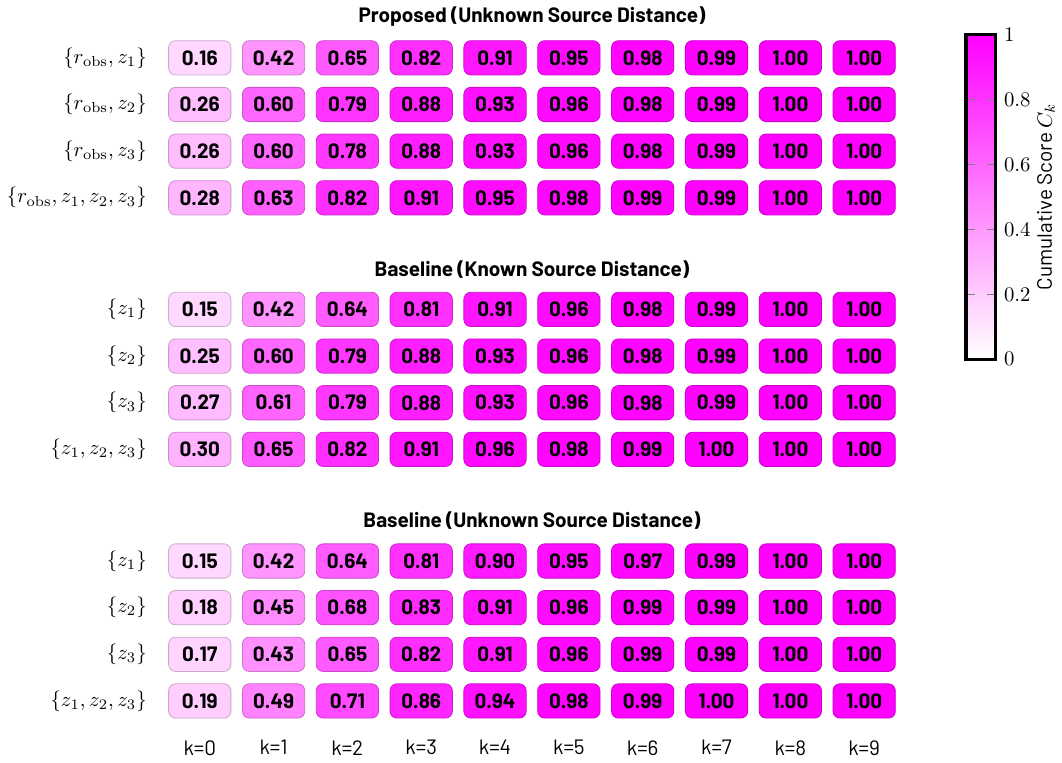}
    \caption{\textbf{Strength estimation performance of learning-based estimation schemes.} \textit{Top:} Cumulative score for several features combined with $r_\mathrm{obs}$. \textit{Middle:} Cumulative score for several features in the baseline scenario with known distances. \textit{Bottom:} Cumulative score for several features in the baseline scenario with unknown distances.}
    \label{fig:eval:strength:matrix}
\end{figure}
While Figure~\ref{fig:eval:strength:main} revealed that all proposed estimators and baseline schemes achieve good strength estimation performance, Figure~\ref{fig:eval:strength:matrix} investigates which features specifically are responsible for this performance. 
To this end, we plot the \ac{CS} achieved by learning-based estimators when using intensity-based features $z_1$-$z_3$ together with $r_\mathrm{obs}$ (top), when the source distance is known in the baseline scenario (middle), and when the source distance is unknown in the baseline scenario (bottom). 
The setting with known source distance and $\{z_1, z_2, z_3\}$ serves as genie-aided reference scheme which is not subject to the ambiguity between distance and source strength.
Similar to our observations in Figure~\ref{fig:eval:strength:main}, the proposed learning-based scheme utilizing $\{r_\mathrm{obs},z_1, z_2, z_3\}$ tracks the \ac{CS} of the genie-aided scheme closely whereas the baseline schemes without knowledge of the source distance have, for lower $k$, a significantly lower \ac{CS}. 
This indicates that at least some informative feature of the source distance, e.g., as provided by $r_\mathrm{obs}$, is required for good source strength estimation performance. When such a feature is provided, the source strength estimation performance is mainly limited by the noise of the intensity-based features and would require longer observation intervals for improved performance.
Additionally, Figure~\ref{fig:eval:strength:matrix} reveals that $z_2$ and $z_3$ are more informative about the source strength than $z_1$ in this intensity-noise-limited regime. Since $z_2$ is the average whiff intensity and $z_3$ is the average whiff intensity slope, i.e., the relationship between $z_2$ and $z_3$ is approximately linear, it is expected that $z_3$ is similarly informative about $J^{(m)}$ as $z_2$. 

\subsubsection{Source Strength Estimation Performance for Different Observation Intervals}\label{sec:eval:strength:obs_interval}
\begin{figure}
    \centering
    \includegraphics[width=\linewidth]{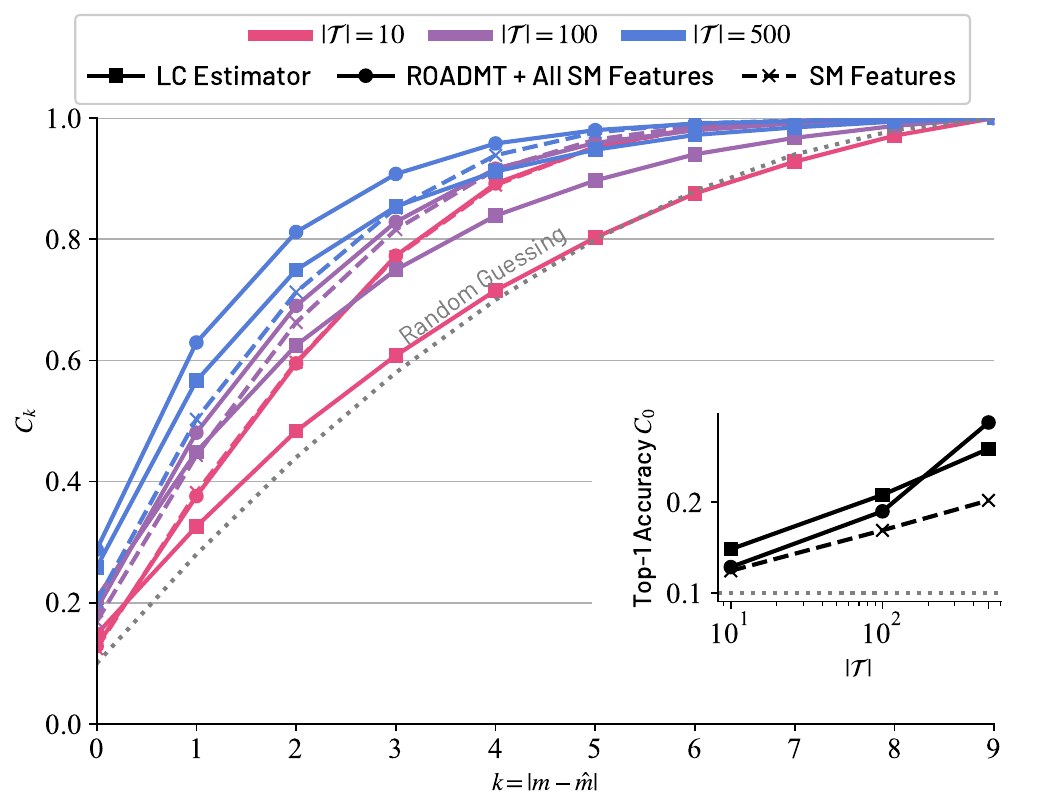}
    \caption{\textbf{Strength estimation performance for different observation windows.} We show the \ac{CS} value of the \ac{LC} estimator (solid, square), the learning-based estimator leveraging $r_\mathrm{obs}$ and $z_1$-$z_6$ (solid, round), and the baseline scheme utilizing only $z_1$-$z_6$ (dashed, cross). The length of the observation window is indicated by the color. For reference, we show the performance of random guessing (dotted). As inset, we additionally show the top-$1$ accuracy as a function of $|\mathcal{T}|$.}
    \label{fig:eval:strength:observation_length}
\end{figure}
Next, we investigate how the duration of the observation interval $\mathcal{T}$ affects the accuracy of the strength estimation. To this end, we show in Figure~\ref{fig:eval:strength:observation_length} for $|\mathcal{T}| \in \{10, 100, 500\}$ the \ac{CS} for the \ac{LC} estimator (solid, squares) and the learning-based estimator (solid, circles) utilizing $r_\mathrm{obs}$ and $z_1$-$z_6$. As reference, we show the baseline scheme with unknown source distances (dashed, crosses).
As expected, the performance of all shown schemes increases with $|\mathcal{T}|$. This is true for the top-1 accuracy, $C_0$, which is shown as inset, but also the other \ac{CS} values. 
Most likely due to the high distance estimation error for short observation intervals (see Figure~\ref{fig:eval:distance:interval-length}), the \ac{LC} classifier performs only slightly better than random guessing (grey dotted) for $|\mathcal{T}|=10$. However, it always achieves a higher top-1 accuracy $C_0$ than the baseline scheme, and as $|\mathcal{T}|$ increases, it achieves similar performance as the baseline scheme for all \ac{CS} values. The proposed learning-based strength estimator has the highest \ac{CS} values for $k>0$ for any given interval length. For small interval lengths, it has a lower top-1 accuracy $C_0$ than the \ac{LC} estimator. We attribute this to the proposed loss function that encourages learning-based schemes to trade higher \ac{CS} values for $k>0$ for a slightly lower top-1 accuracy as we will show in the following. 

\subsubsection{Impact of $\lambda$ on Strength Estimation Performance}
\begin{figure}
    \centering
    \includegraphics[width=\linewidth]{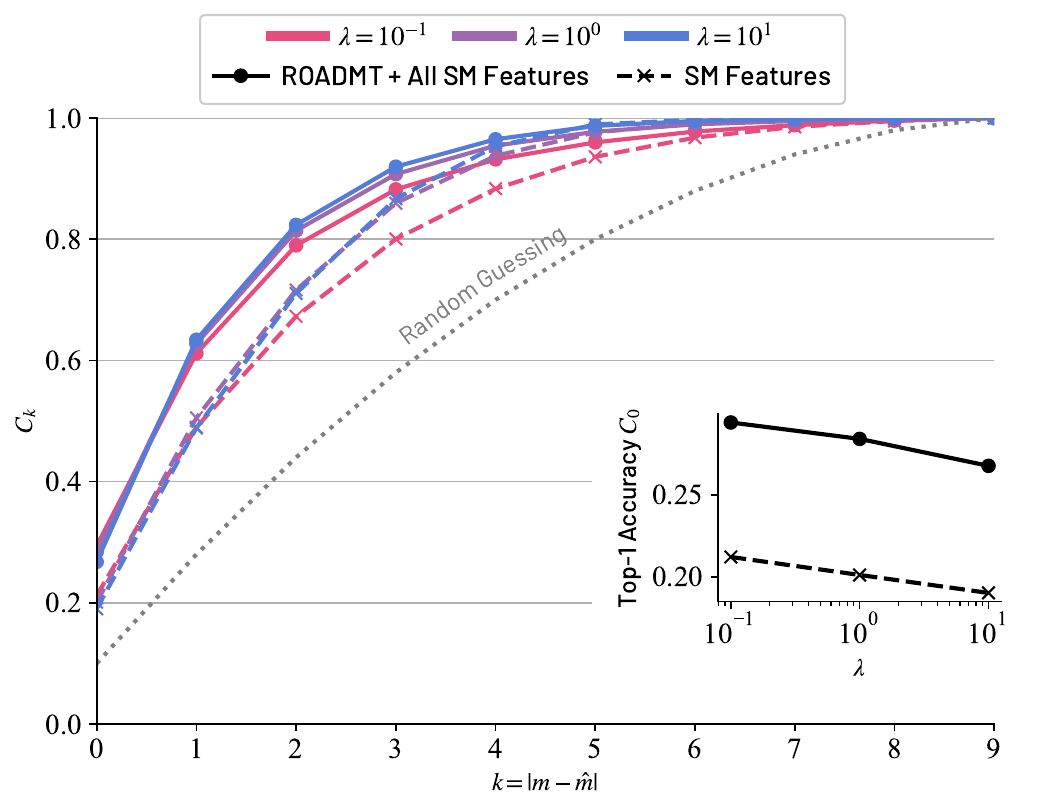}
    \caption{\textbf{Strength estimation performance for different $\lambda$.} We show the \ac{CS} values for various weighting factors $\lambda$ when employing the \ac{ROAMT} and all \ac{SM} features (solid) compared to employing only the \ac{SM} features (dashed). As inset, we additionally show the accuracy $C_0$ as a function of $\lambda$.}
    \label{fig:eval:strength:weighing-factor}
\end{figure}
Finally, we investigate in Figure~\ref{fig:eval:strength:weighing-factor} how the weighting factor $\lambda$, which determines the relative importance of the $\mathrm{CrossEntropy}$ loss and the loss punishing larger errors $|\hat{m}-m|$, affects the performance of the learning-based estimators. 
For visual clarity, we show only results for the proposed learning-based estimator leveraging \ac{ROAMT} and all \ac{SM} features and the baseline estimator leveraging all \ac{SM} features when the distance is unknown. 

First, we observe in the inset of Figure~\ref{fig:eval:strength:weighing-factor} that the accuracy, i.e., $C_0$, decreases as $\lambda$ increases. This is expected since increasing $\lambda$ decreases the relative importance of the cross-entropy loss in~\eqref{eq:strength:learning:loss}, which trains the network to focus only on correct vs. incorrect decisions. 
However, as $\lambda$ increases, the \ac{CS} for larger values of $k$ increases, especially for the baseline scheme. Nevertheless, also the proposed learning-based scheme has higher \ac{CS} values for $k \geq 1$ as $\lambda$ increases. 
This demonstrates that $\lambda$ enables a smooth trade-off between optimizing top-$1$ accuracy and estimation with good average performance. 

\textbf{Summary}: In this section, we considered source strength estimation and showed that \ac{ROAMT} can effectively resolve the ambiguity between source distance and source strength. We identified the average whiff intensity, $z_2$, as most informative feature for source strength and proposed an \ac{LC} strength estimator based on $z_2$ and $r_\mathrm{obs}$ that enables source strength estimation even in data-scarce settings.  

\section{Conclusion}\label{sec:conclusion}
In this paper, we investigated how the degradation diversity of molecule mixtures can be leveraged for the estimation of the distance to a molecule source and the source strength in challenging environments subject to turbulent flow. 
Because different molecule types have different degradation rate constants, the \ac{ROAMT} reflects the molecules' travel time and consequently the expected distance to the source. 
First, we proposed an \ac{LC} distance estimator requiring only knowledge of three parameters of the molecule source and the propagation environment, all of which can be readily measured. For settings requiring higher estimation accuracies, we proposed a learning-based distance estimator combining \ac{ROAMT} with intensity- and timing-based features established in the literature. 
Based on simulations with millions of molecules in complex turbulent airflow, we showed that the \ac{LC} estimator can approach the performance of the learning-based baseline scheme that exploits only intensity- and timing-based features of a single molecule type. The proposed learning-based estimator leveraging \ac{ROAMT} can outperform this baseline systematically both in settings with known and unknown source strength. 

Additionally, we proposed an \ac{LC} estimator and a learning-based scheme for the source strength, i.e., to determine how many molecules the source releases per unit time. 
Both of these schemes leverage \ac{ROAMT} to resolve the ambiguity between source strength and distance to the source. 
For the same simulation scenario as for distance estimation, we showed that, for sufficiently long observation windows, the \ac{LC} strength estimator can achieve a performance comparable to a learning-based baseline scheme leveraging only intensity- and timing-based features of a single molecule type. 
On the other hand, the proposed learning-based strength estimator exploiting \ac{ROAMT} outperforms the aforementioned baseline scheme if the source distance is unknown. 
Notably, it closely tracks the performance of a baseline scheme that knows the true source distance and is therefore not subject to the ambiguity between distance and source strength.

Building on these promising results for both distance and strength estimation, important directions for future research include investigating degradation diversity in mixtures with three or more species and the experimental validation of the proposed schemes.

\bibliographystyle{ieeetr}
\bibliography{literature}
\clearpage
\end{document}